\documentclass[a4paper,titlepage,11pt]{report}
\usepackage[utf8x]{inputenc}
\usepackage[T1]{fontenc}
\usepackage{latexsym}

\usepackage{graphicx}
\usepackage{comment}

\usepackage{setspace}
\usepackage{amsmath}
\usepackage{amssymb}
\usepackage{mathrsfs}
\usepackage{bm}
\usepackage{amsthm}

\theoremstyle{definition}
\newtheorem{definition}{Definition}[chapter]
\theoremstyle{plain}
\newtheorem{theorem}{Theorem}[chapter]
\newtheorem{lemma}{Lemma}[chapter]

\title{}
\author{}
\date{}

\begin{document}

%%%%%%%%%%%%%%%%%%%
%\begin{titlepage}
%\begin{flushright}
%LU-TP 11-28 \\
%September 2011
%\end{flushright}

\begin{center}
\Large{\textbf{The entanglement or separability of \emph{mixed} 
quantum\\ states as a matter of the choice of
observables\\}}
\vspace{.5cm}
\large{Iacopo Pozzana\footnote{i.pozzana@gmail.com}}\\ 
\vspace{.3cm}
Dipartimento di Fisica, Universit\`a di Pisa \\
Largo Pontecorvo 3, 56127 Pisa, Italy

\end{center}
\noindent
In quantum systems, entanglement corresponds to nonclassical correlation of nonlocal observables.
Thus, entanglement (or, to the contrary, separability) of a given quantum state is not uniquely
determined by properties of the state, but may depend on the choice of the factorization of the algebra
of observables. In the present work, we expose and systematize some recently reported results about
the possibility to represent a single quantum state as either entangled or separable.
We will distinguish in particular the cases of pure and mixed states. For pure states,
it has been shown that observables can always be constructed such that any state has any amount
of entanglement possible. For mixed states, the situation is more complex and only partial results
are known: while it is always possible to choose a factorization such that a state appears separable,
a general criterion to determine whether a state can be represented as entangled is not known.
These results will be illustrated by several examples, the phenomenon of quantum teleportation,
and the geometry of the states of two qubits.

\begin{center}

\vskip 1.5cm 
\textbf{\Large{Bachelor Thesis\\}}

\vspace{.5cm}
Submitted to:\\
Dipartimento di Fisica, Universit\`a di Pisa

\vspace{.5cm}
Defended on:\\
29 February 2012

\vspace{.5cm}
\large{Advisor:\\
Hans-Thomas Elze\footnote{elze@df.unipi.it}}

\end{center}

%\cleardoublepageddddd
%%%%%%%%%%%%%%%%
\begin{abstract}

In this work we will examine how a given quantum state can appear to be either entangled or separable depending on the choice
of the factorization of the algebra of observables of a composite multi-partite
quantum system.~\footnote{The main reference on which this work will be based is 
\cite{thirring}.}
We will focus our attention on mixed states here. 

For pure states it is always possible to find both a factorization such that the state
appears entangled and another factorization where it appears separable.
On the other hand, for a mixed state it is not always possible to find a factorization such that it appears maximally entangled.
We will discuss this statement and investigate under what
conditions a mixed state can be represented as entangled.

In the first part of this work, we will introduce some basic concepts necessary for the following discussion,
such as the Hilbert space, the tensor product structure, the qubit and the density operator.

In the central part, some possible constraints under wich a quantum state can be represented as entangled are discussed:
\begin{itemize}
\item For pure states no further restrictions are necessary and we will summarize the relevant results reported in the
literature recently. In short, depending on the choice of observables,
a pure quantum state of a composite system can be demonstrated to be either entangled or
separable. In this sense, entanglement is not an invariant property of a given state,
despite its fundamental role concerning properties of multi-partite systems
(such as in the famous Einstein-Podolsky-Rosen considerations).
\item For mixed states the situation is more complex and only partial results are known.
Indeed, it is always possible to choose a factorization such that a state appears separable,
as we shall demonstrate. However, it is not known in general how much ``mixedness'' is
admissible, such that a state can still be represented as entangled. This will be
investigated in more detail in our work.
\end{itemize}

The general statements at which we arrive will be applied to some illustrative examples and to the well
known phenomenon of quantum teleportation, an entanglement feature of great interest in quantum information theory.
We will also describe and discuss the geometry of quantum physical states for the case of two qubits.

We will close with our conclusions, an outlook on possible extensions of the results, 
and a list of references to the literature.
\thispagestyle{empty}
%\cleardoublepage
\thispagestyle{empty}
\end{abstract}

\tableofcontents

%%%%%%%%%%%%%%%%%%%%%%
\chapter{Introduction}
\begin{quotation}
\begin{flushright}
 ``If we have to go on with these damned quantum jumps, \\ then I'm sorry that I ever got involved.'' \\
 \vspace{0.2cm} \emph{E. Schr\"{o}dinger}
\end{flushright}
\end{quotation}
\vspace{1cm}

When two (or more) quantum systems interact, for example in a scattering process, a correlation arises
among them; this correlation lasts in time and can be exploited to collect information about, or even modify,
one of the systems without directly acting on it, no matter how far it is.

This phenomenon, which is called \emph{entanglement} and has no analogoue in classical theory,
grossly violates the locality principle,
as was first highlighted by Einstein, Podolski and Rosen in \cite{epr}.
Purpose of the authors was to demonstrate that the existence of entanglement is incompatible
with the completeness of quantum theory, but their thesis was proved to be wrong in many occasions
over time, first of all by John S. Bell in his 1964 article ``\emph{On the Einstein-Poldolsky-Rosen paradox}''
\cite{bell}.
Bell showed that any theory reproducing the predictions of quantum mechanics has to violate locality.

Quantum theory is thus a complete non-local theory, with entanglement as an inevitable consequence of its
principles. The article by EPR, against the scopes of the authors, laid the bases for the birth
of quantum information theory, the purpose of which is to study the nature of entanglement and the way to
exploit it.

Let us consider a state in a composite quantum system. It may seem reasonable to suppose that this state has a certain
fixed amount of entanglement, or in other words, the system has a fixed amount of quantum correlation within.
However, it has been proven, as we are going to explain in more detail in the present work,
that the quantity and even the existence of
entanglement for a certain state depends on the choice of the factorization of the algebra of observables,
i.e.~on the choice of the physical observables used to describe it. As an example, a composite system
consisting of a couple of spins could be described in terms of total spin, and in this factorization no
entanglement will appear: the system would not even look ``composite''.

The possibility for entanglement to be ``created'' and ``destroyed'' is maybe its most fascinating feature,
and is the basis for many of the recent achievements in quantum information and quantum computing,
such as the protocols of quantum cryptography and quantum teleportation, or the superdense coding.
The questions ``when'' and ``how'' it is possible to manipulate entanglement are thus of main importance.

The first part of the present work will be devoted to a summary of the necessary basic concepts
to deal with quantum mechanics and, in particular, with quantum information theory, such as the Hilbert space,
the tensor product space, the density operator, and so on. \\
Then, in {\bf Chapter \ref{chap2}}, we will start the central part of our exposition,
treating the case of pure states; it will be shown how it is always possible to make any pure state appear either
entangled or not (in the latter case we talk of \emph{separable} or \emph{factorized} states).\\
The possibility for a state to be described in terms of a wavefunction is strictly connected
with the possibility to make it appear entangled (or not). In {\bf Chapter \ref{chap3}}
we will thus consider the more general case of mixed states, and we will see that the situation here is more complex:
we will demonstrate that it is still always possible to represent any state as separable, but we will also
see examples of \emph{absolutely separable} states, i.e.~states which appear separable in every factorization;
even if it is not always possible to make a mixed state appear entangled, we will introduce two classes
of states for which this possibility exists. The chapter will be closed by two complements, a basic exposition of the protocol
of quantum teleportation, and a description of the geometry of physical states for the case of two qubits,
i.e.~for $2 \times 2$ dimensional composite system. \\
We will then draw our conclusions, briefly summarizing the results here exposed. The present work includes,
in its last pages, a short list of references to the literature.

%%%%%%%%%%%%%%%%%%%%%%%%
\chapter{Basic concepts}

%%%%%%%%%%%%%%%%%%%%%%%
\section{Hilbert space} 

A \emph{Hilbert space} is a vector space, real or complex, where an inner product is defined and which is a complete metric 
space with respect to the norm induced by this inner product. A Hilbert space is called \emph{separable} 
if it admits a countable complete orthonormal basis set.

In quantum mechanics a physical system is associated with a separable Hilbert space ${\cal H}$.
The state of the system is represented by a ray in ${\cal H}$.
Using the Dirac notation we will denote a vector in ${\cal H}$ by $|\psi \rangle$ and its dual by $\langle \psi|$.
Correspondingly, given any two vectors $|\psi \rangle$ and $|\phi \rangle$, their inner product 
and the associated norm will be respectively denoted 
by $\langle \phi | \psi \rangle $ and $ || \psi || = \sqrt{ \langle \psi | \psi \rangle } $.

Since, for any complex $\alpha$, the two vectors $|\psi \rangle$ and $\alpha |\psi \rangle$ correspond to the 
same state, we can always represent a state by a vector with unit norm (and so we'll do, except where otherwise specified).

%%%%%%%%%%%%%%%%%%%%%%%%%%
\subsection{Complete sets} 

Given a complete set of vectors $\{ |n \rangle \}_{n=1,2,3, \ldots}$, the existence of which 
is granted by the 
separability, we can expand any other vector $|\psi \rangle$ as:
$$
|\psi \rangle = \sum_n c_n |n \rangle , 
$$
where $c_n = \langle n | \psi \rangle $ is the projection of $|\psi \rangle$ along $|n \rangle$.
The completeness of the set $|n \rangle$ is expressed by the \emph{closure relation}
$$
\sum_n c_n |n \rangle \langle n| = \hat{\mathbb I},
$$
where $\hat{\mathbb I}$ represents the identity operator.

Another useful property that we can require of the set of basis states is orthonormality:
$$
\langle n|m \rangle = \delta_{nm} ,
$$
where $\delta_{nm}$ is the Kronecker delta.

%%%%%%%%%%%%%%%%%%%%%%%% 
\subsection{Observables} 

Within this representation, a physical observable is represented by a hermitian 
operator (say $\hat{A}$) 
which has a complete set of eigenvectors (\emph{eigenstates}).
If the system is in an eigenstate of $\hat{A}$ (say $|n \rangle$), then the result of a 
measurement of the observable is certain: it corresponds to the eigenvalue of 
$\hat{A}$ associated with $|n \rangle$, i.e.:
$$
\hat{A}|n \rangle = A_n|n \rangle. 
$$

For a generic state $|\psi \rangle$ it is not possible to make a certain 
prediction of the outcome of a measurement. Instead,
we can define an \emph{expectation value} of $\hat{A}$ on $|\psi \rangle$ as
\begin{equation} \label{expect}
\langle \hat{A} \rangle \equiv \langle \psi|\hat{A}|\psi \rangle,
\end{equation}
which is the statistical mean value of the operator on the state.

An operator acting on a separable space can be represented as a matrix. 
Its matrix elements with respect to an assumed complete basis are given 
by:
$$
{\hat A}_{mn} = \langle m| {\hat A}|n \rangle.
$$

%%%%%%%%%%%%%%%%%%%%%%%%%%%%%% 
\section{Tensor product space} 

Let us consider a system which is composed of a number $N$ of subsystems. 
Since, under certain circumstances, every subsystem can be studied separately, 
it must be represented by a Hilbert space.
In our representation, then, the relation between the total system and the subsystems 
can be simply expressed in terms of a tensor product:
$$
{\cal H} = \bigotimes^N_{i = 1} {\cal H}_i,
$$
where ${\cal H}_i$ and ${\cal H}$ denote the space of the $i$-th subsystem and 
the total system, respectively.~\footnote{Let us emphasize that the values 
assumed by the parameter $i$ are just labels, not necessarily numbers.}
Thus, ${\cal H}$ is said to have a \emph{tensor product structure}(TPS).

Given a vector $| \psi_i \rangle$ for each of the subsystems, we can build a vector of the total system:
$$
|\Psi \rangle = \bigotimes^N_{i = 1} |\psi_i \rangle.
$$
For an easier reading we can use the notation $|\Psi \rangle = |\psi_1 \rangle |\psi_2 \rangle \cdots |\psi_N \rangle$ or, when no confusion is possible, 
the even more compact $|\Psi \rangle = |\psi_1\ \psi_2\ \cdots\ \psi_N \rangle$. 

Notice that not every vector in ${\cal H}$ can be written in this \emph{factorized form}; in quantum physics factorizability has a precise meaning, 
corresponding to the possibility that a state is completely described in terms 
of the states of the subsystems, which is, of course, not always possible. 
We shall discuss related issues further in due course. 

%%%%%%%%%%%%%%%%%%%%%%%%%%%%%% 
\subsection{Factorized states} 

Given the sets of basis states $\{|n_1\rangle\},\{|n_2\rangle\},\cdots,\{|n_N\rangle\}$ for 
the subsystems "1", "2", etc., it is then trival to demonstrate that $\{ |n_1\ n_2\ \cdots\ n_N\rangle\}$ 
forms a basis for ${\cal H}$. 
We can see that the dimension $d$ of ${\cal H}$ is given by the product of 
the dimensions $d_i$ of its subspaces:
$$
d= \prod^N_{i=1} d_i.
$$
Orthonormality is preserved:
\begin{eqnarray*}
\langle n_1'\ n_2'\ \cdots\ n_N'|n_1''\ n_2''\ \cdots\ n_N''\rangle & = & \langle n_1'|n_1''\rangle\ \langle n_2'|n_2''\rangle\ \cdots\ \langle n_N'|n_N''\rangle \\ 
& = & \delta_{n_1'n_1''}\ \delta_{n_2'n_2''}\ \cdots\ \delta_{n_N'n_N''}\ .
\end{eqnarray*}

Let us expand a factorized $|\Psi_f \rangle \in {\cal H}$ on this new set:
\begin{eqnarray*}
|\Psi_f \rangle & = & \bigotimes^N_{i = 1} |\psi_i \rangle\ \\
& = & ( \sum_{n_1} c_{n_1} |n_1 \rangle ) \cdot\ ( \sum_{n_2} c_{n_2} |n_2 \rangle ) \cdot\  (\cdots) \cdot\ ( \sum_{n_N} c_{n_N} |n_N \rangle ) \\
& = & \sum_{n_1,n_2,\cdots,n_N}c_{n_1}\, c_{n_2} \cdots c_{n_N} \,|n_1 \rangle \,|n_2 \rangle \cdots \,|n_N \rangle.
\end{eqnarray*}
Similarly, a generic $|\Psi \rangle$ can be written as:
\begin{equation} \label{Psi}
|\Psi \rangle = \sum_{n_1,n_2,\cdots,n_N}c_{n_1, n_2, \cdots, n_N} |n_1 \rangle \,|n_2 \rangle \, \cdots \,|n_N \rangle,
\end{equation}
but now each coefficient depends, in general, on all the subspaces. This is the mathematical 
expression of the fact that a composite quantum mechanical system, generally, cannot 
be completely characterized in terms of its subsystems: there exist more general states than 
the factorized ones.

%%%%%%%%%%%%%%%%%%%%%%%%%%%%%%%%%%%%% 
\subsection{Tensor product operators} 

An operator on ${\cal H}_1$, for example, can be extended to the entire ${\cal H}$ by just setting 
it equal to the identity with respect to the other subspaces:
\begin{eqnarray*}
{\hat A}^{(1)} : {\cal H}_1 \rightarrow {\cal H}_1, & \hat{\mathbb I}^{(i)} : {\cal H}_i \rightarrow {\cal H}_i, & {\hat A} : {\cal H} \rightarrow {\cal H}, \\
& {\hat A}  \equiv  {\hat A}^{(1)} \otimes \hat{\mathbb I}^{(2)} \otimes \cdots \otimes \hat{\mathbb I}^{(N)}. &
\end{eqnarray*}
The eigenspaces of the extended operator are formed by all those vectors whose projection on ${\cal H}_1$ is an eigenvector for ${\hat A}^{(1)}$; 
for every eigenvector of ${\hat A}^{(1)}$ we now have $d$ independent eigenvectors of ${\hat A}$, where $d$ 
is the dimension of $\otimes^N_{i=2}{\cal H}_i$ (eventually a countable infinity), with the same eigenvalue: 
the spectrum of the eigenvalues does not change but the degeneracy is multiplied by $d$.

Generally speaking, we can combine $N$ operators acting on the subspaces in one operator acting on ${\cal H}$.
We use the tensor product structure as for the vectors:
\begin{equation} \label{facop}
{\hat A} = \bigotimes^N_{i=1} {\hat A}^{(i)}
\end{equation}
and, as for the vectors, we point out that not every operator can be written in the factorized form (\ref{facop}).
Because of the linearity of the various ${\hat A}^{(i)}$, tensor product operators act on a generic $|\Psi \rangle \in {\cal H}$ in the 
folllowing way:
$$
{\hat A} |\Psi \rangle = \sum_{n_1,n_2,\cdots,n_N}c_{n_1, n_2, \cdots, n_N} ({\hat A}^{(1)}|n_1 \rangle) \,({\hat A}^{(2)}|n_2 \rangle) \, 
\cdots \,({\hat A}^{(N)}|n_N \rangle),
$$
so that if the state is factorized the expression simplifies to become 
\begin{eqnarray*}
{\hat A} |\Psi_f \rangle &=& \bigotimes^N_{i = 1} {\hat A}^{(i)}|\psi_i \rangle \\
&=& \prod^N_{i=1} ( \sum_{n_i} c_{n_i} {\hat A}^{(i)}|n_i \rangle ).
\end{eqnarray*}

%%%%%%%%%%%%%%%%%%%%%%%%%%%%%%%%%%%%%%%%%%%%%%%%%%%%%%%% 
\section{Density operator, pure states and mixed states} 

Sometimes is not possible to give a complete description of a quantum system. 
That may happen:
\begin{itemize}
\item because we have access only to a part of the system, for example, 
when we cannot observe all the products of a decay process, 
so that a complete description is \emph{actually impossible};
\item because the system is too wide or complicated, as a macroscopic system, 
making a complete description \emph{practically impossible};
\item because we are only interested in one part of the system: 
a complete description is \emph{unwanted}.
\end{itemize} 

To begin the discussion with the simplest cases of a \emph{bi-partite system}, 
consisting of subsystems $A$ and $B$, 
we imagine to divide the total Hilbert space ${\cal H}$ into 
two subspaces: ${\cal H}_A$, to which we do have access, and ${\cal H}_B$, 
to which we do not.~\footnote{Even without further information about subsystem $B$, 
we assume in the following that its Hilbert space ${\cal H}_B$ is separable.}
However, in general, the state of the composite system cannot be factorized 
in the form $|\Psi \rangle = |\psi_A \rangle \otimes |\psi_B \rangle $, i.e.,  
we cannot describe completely the subsystem $A$ (or $B$) alone. 

Nevertheless, we are interested to evaluate, for example, the expectation values of
the observables acting on ${\cal H}_A$ using the information we have.

%%%%%%%%%%%%%%%%%%%%%%%%%%%%% 
\subsection{Density operator} 

The expectation value of an operator ${\hat A}: {\cal H}_A \rightarrow {\cal H}_A$, 
as defined in (\ref{expect}), recalling also (\ref{Psi}), is:
\begin{eqnarray*}
\langle \hat{A} \rangle &\equiv& \langle \Psi|\hat{A}|\Psi \rangle \\
&=& \sum_{a,a',b,b'} \langle \Psi |a',b' \rangle \langle a',b'| \hat{A} | a,b \rangle \langle a,b| \Psi \rangle \\
&=& \sum_{a,a',b} \langle \Psi |a',b \rangle \langle a'| \hat{A} | a \rangle \langle a,b| \Psi \rangle.
\end{eqnarray*}
We have used the closure relation $\sum_{a,b} | a,b \rangle \langle a,b| = \hat{\mathbb I}$ 
in the second line and, in the third, the equality 
$$
\langle a,b| \hat{A} | a',b' \rangle = \langle a| \hat{A} | a' \rangle \langle b|b' \rangle = \langle a| \hat{A} | a' \rangle \delta_{bb'},
$$
due to the fact that $\hat{A}$ acts only on ${\cal H}_A$.

Here it is appropriate to introduce the \emph{density operator}:
\begin{equation} \label{rho}
\hat{\rho}_{aa'} \equiv \sum_b \langle a,b| \Psi \rangle \langle \Psi| a',b \rangle, 
\end{equation}
i.e., by its matrix elements; this may require additional considerations 
if non-separable spaces are admitted, which we presently do not pursue, for simplicity.
With this powerful tool the above expectation value of observable $\hat{A}$ 
can be expressed in the very elegant way
$$
\langle \hat{A} \rangle = \mbox{Tr}(\hat{\rho} \hat{A}),
$$ 
where the trace is over all states of the composite system.  

The density operator has the following three properties:
\begin{enumerate} \label{density}
 \item Normalized: $\mbox{Tr}(\hat{\rho})=1$;
 \item Hermitian: $\hat{\rho}=\hat{\rho}^{\dagger}$;
 \item Positive semi-definite: $x^* \hat\rho\, x \geq 0, \forall x \in{\mathbb C}^N$.~\footnote{Actually, for the density operator as we 
introduced it, holds the even stronger property $\hat{\rho}_{aa'} \geq 0$, which implies the positive-semidefiniteness.}
\end{enumerate}
The demonstrations are trivial, given the explicit form of $\hat\rho$ above. In general, 
these properties are considered to \emph{define} a density operator. 

%%%%%%%%%%%%%%%%%%%%%%%%%%%%%%%%%%
\subsection{Schmidt decomposition}

Consider a bipartite system ${\cal H} = {\cal H}_A \otimes {\cal H}_B$ and the two bases $\{ |a\rangle \}$ for ${\cal H}_A$
and $\{ |b\rangle \}$ for ${\cal H}_B$.

First, let us define the \emph{reduced density operator} for a subsystem, which is given by the partial trace of 
$\hat\rho$ over the other subsystems; for example, in this case we have
$$
\hat\rho^{(A)} \equiv \mbox{Tr}_B \, \hat\rho.
$$

The most general state $|\psi\rangle$ for the system ${\cal H}$ and the associated density operator are:
\begin{align*}
|\psi\rangle &= \sum_{a,b} c_{ab} |a\rangle |b\rangle,\\
\hat\rho &= \sum_{a,b,a',b'} c_{ab} c^*_{a'b'} |a\rangle |b\rangle \langle a'| \langle b'|.
\end{align*}

The reduced density operators for the subsystems are:
\begin{equation} \label{reduced}
\begin{split}
 \hat\rho^{(A)} &= \mbox{Tr}_B \, \hat\rho = \sum_b \langle b| \hat\rho |b \rangle \\
&= \sum_b \langle b| \left( \sum_{a,b,a',b'} c_{ab} c^*_{a'b'} |a\rangle |b\rangle \langle a'| \langle b'| \right) |b \rangle \\
&= \sum_{a,a',b} c_{ab} c^*_{a'b} |a\rangle \langle a'|, \\
 \hat\rho^{(B)} &= \sum_{a,b,b'} c_{ab} c^*_{ab'} |b\rangle \langle b'|. \\
\end{split}
\end{equation}

We can choose, for the $\{ |a\rangle \}$, an orthonormal set of eigenvectors for $\hat\rho^{(A)}$.
Then, defining \emph{Schmidt coefficients} $C_a \equiv \sum_b |c_{ab}|^2$ (eigenvalues of $\hat\rho^{(A)}$),
from (\ref{reduced}) we have:
\[
 \hat\rho^{(A)} = \sum_{a,b} |c_{ab}|^2 |a\rangle \langle a| = \sum_a C_a |a\rangle \langle a|.
\]
We can also introduce the normalized vectors $|B_a\rangle \equiv \frac{1}{\sqrt{C_a}} \sum_b c_{ab} |b\rangle \in {\cal H}_B$,
which allow us to write:
\begin{align*}
 |\psi \rangle &= \sum_a \sqrt{C_a} |a\rangle |B_a\rangle, \\
 \hat\rho &= \sum_a C_a |a\rangle |B_a\rangle \langle B_{a}| \langle a|, \\
 \hat\rho^{(B)} &= \sum_a C_a |B_a \rangle \langle B_a|.
\end{align*}
So we see that the $|B_a\rangle$ are eigenvectors for $\hat\rho^{(B)}$: the two reduced density operators
are simultaneously diagonal and they have the same eigenvalues $C_a$.

The procedure used above is called \emph{Schmidt decomposition} and allows us to express a pure state
and its density operator in terms of a sum over one index only. The sum is limited by the dimension
of the smallest of the subspaces, in fact we have implicitly assumed $\dim {\cal H}_A \leq \dim {\cal H}_B$.

When a state is expressed in its Schmidt form, it is straightforward to determine if it is factorized or not:
if only one Schmidt coefficient is nonvanishing, the state is factorized; otherwise it is \emph{entangled} (not factorized).

%%%%%%%%%%%%%%%%%%%%%%%%%%%%%%%%%% 
\subsection{Pure and mixed states} 

Those $\langle a,b|\Psi \rangle$ appearing (with their complex conjugates) in (\ref{rho}) are nothing but
the coefficients of the expansion of $|\Psi \rangle$ with respect to the basis states of the form $|a\rangle \otimes |b\rangle$. 
If the state is factorized, $|\Psi\rangle = |\psi_A \rangle \otimes |\psi_B \rangle$, the coefficients depend on only one
subspace each. Then, defining $c_a\equiv\langle a|\Psi_A\rangle$, 
$c_b\equiv\langle b|\Psi_B\rangle$, we obtain from (\ref{rho}):
\begin{eqnarray*}
 \hat\rho_{aa'} &=& \sum_b c_a c_b c^*_{a'}c^*_b = c_a c^*_{a'} \sum_b |c_b|^2 = c_a c^*_{a'} \\
 \Rightarrow \qquad \hat\rho &=& |\psi_A \rangle \langle \psi_A|,
\end{eqnarray*}
because, even if we do not know the $c_b$, the closure relation still holds for 
the basis $\{ |b \rangle \}$.

Furthermore, we can choose a complete set for ${\cal H}_A$ including $|\psi_A \rangle$ 
to obtain $\hat\rho$ in diagonal form:
$$
\hat\rho = \left(
\begin{array}{cccccc}
 0            \\
 & \ddots     \\
 && 0         \\
 &&& 1        \\
 &&&& 0       \\
 &&&&& \ddots \\
\end{array} \right).
$$
If we know the wavefunction $|\psi_A \rangle$, 
which represents the state of the subsystem to which we have access, 
i.e., if the state of the total system ${\cal H}$ is factorized 
in terms of states of ${\cal H}_A$ and ${\cal H}_B$, we can say our information is complete, 
since the uncertainty, if any, about the outcome of measurements 
on $A$ does not depend on our ignorance. 

To identify these states, 
which we call \emph{pure states}, we introduce the following 
important property of the density operator:
$$
\hat\rho = \hat\rho^2 \quad \Longleftrightarrow \quad|\psi_A \rangle \mbox{ is a pure state}.
$$ 
That is, the density operator corresponding to a pure state is a \emph{projector}; 
in particular, we have $\hat\rho =|\Psi_A\rangle\langle\Psi_A|$ in the present case. 

We call \emph{mixed state} a generic state which, generally, may not be 
factorizable.~\footnote{According to this definition a pure state is a mixed state too.} 
In the latter case, we do not have a wavefunction on ${\cal H}_A$ and the system cannot be described but in terms of the density operator.

A situation of particular interest occurs when our missing information is due to 
the statistical nature of the system.
Let $\{ |\psi^{(i)} \rangle \}$ be a set of possible states for our system. 
Suppose we know the probability $w_i$ for the system 
to be in the $i$-th state.~\footnote{This probability is related to a statistical 
approach in studying the system and has nothing to do
with the intrinsic statistical (indeterministic) nature of quantum mechanics itself.}
If $\{ |\psi_n \rangle \}$ is a complete set, we can write:
$$
|\psi^{(i)} \rangle = \sum_n a^{(i)}_n |\psi_n \rangle.
$$
For the expectation value of an observable we have:
$$
\langle {\hat O} \rangle = \sum_i w_i \langle \psi^{(i)}|{\hat O}|\psi^{(i)} \rangle = \sum_{i,n,m} w_i \, a^{(i)*}_m a^{(i)}_n {\hat O}_{mn}.
$$
Now the density operator assumes the form:
$$
\hat\rho = \sum_i w_i |\psi^{(i)} \rangle \langle \psi^{(i)}|, 
\qquad \hat\rho_{nm} = \sum_i w_i \, a^{(i)*}_m a^{(i)}_n.
$$
Thus, $\hat\rho$ is the weighted (by $w_i$) mean value of the 
$\hat\rho$ of the single states $|\psi^{(i)} \rangle$.
Since $\{|\psi^{(i)} \rangle \}$ is a complete set, 
we can represent the density matrix with respect to it in diagonal form:
$$
\hat\rho = \left(
\begin{array}{cccc}
 w_1        \\
 & \ddots   \\
 && w_n     \\
 &&& \ddots \\
\end{array} \right).
$$

%%%%%%%%%%%%%%%%%%%%%%
\subsection{Bell basis} \label{sbell}
A \emph{qubit} is a two-states quantum system. It is the minimal unit of quantum information.
Typical realizations are an atom with its ground state and one of its excited states, a photon
with its two polarization states (for fixed frequency and direction) or an electronic spin with its
``up'' and ``down'' states.
A commonly used basis for qubit representation is the \emph{Pauli basis}:
$$
\{ |\uparrow \rangle , |\downarrow \rangle \}.
$$

Among all possible bipartite systems, two-qubit systems are of particular importance
for their facility to be treated theoretically and experimentally.
To represent two qubits we shall often use the \emph{Bell basis}:
\begin{equation} \label{bell}
\begin{split}
 | \Phi^\pm \rangle \equiv \frac{| \uparrow \uparrow \rangle \pm |\downarrow \downarrow \rangle}{\sqrt{2}}, \\
 | \Psi^\pm \rangle \equiv \frac{| \uparrow \downarrow \rangle \pm |\downarrow \uparrow \rangle}{\sqrt{2}}.
\end{split}
\end{equation}
The states composing a Bell basis have the peculiarity of being maximally entangled; since entanglement is,
as we shall see, a factorization-dependent property, it is understood
that a Bell state remains entangled as long as we do not change the choice of the decomposition for the
total system.

We shall now clarify what we mean with \emph{maximal entanglement}. A varitey of measures have been introduced
to quantify the amount of entanglement, mostly based on the Shannon or von Neumann
entropy, applied to a density matrix; the interested reader will find more in \cite[Sect. 7.1.3]{diosi} or,
for the two-qubit case, in \cite{ishizara}. Here we shall just notice that the entanglement measure
(\emph{entanglement of formation}) most
commonly used in the literature is a quantity bounded below and above, so maximal
entanglement is a well defined concept for a fixed factorization.

%%%%%%%%%%%%%%%%%%%%%%%%%%%%%%%%%%% 
\subsection{Some brief conclusions} 

Following the presentation in Ref.~\cite{konishi}, we have introduced the density 
operator in a theoretical and general approach.
We have also defined pure and mixed states in relation to the possibility to 
factorize them in a certain representation of a composite system.~\footnote{We did 
not begin with their relation to particular properties of the density operator, 
although this presents an alternative, as we have seen through the differing forms 
of the diagonalized density matrix.} Whether this 
possibility of factorization exists or not is not completely dictated by 
the ``book of nature'', since factorizability is a matter of choice of the 
representation, as we shall see in more detail. 

But, given  a state, its purity does not depend on the set of eigenvectors chosen to span the space, 
i.e.~on the choice of the observables, since the property $\hat\rho = \hat\rho^2$ still holds under basis changes.

In the next chapter the concept of entanglement will be introduced.
Some results from Refs.~\cite{harshman} and~\cite{caltagirone}
about the possibility of finding and creating entanglement for the case of pure states will be exposed.

%%%%%%%%%%%%%%%%%%%%%%%%%%%%%%%%%%%%%
\chapter{Entanglement in pure states} \label{chap2}

%%%%%%%%%%%%%%%%%%%%%
\section{Entanglement and separability} 

Let us consider a bipartite system $S_{AB} = S_A \otimes S_B$ (the generalization to any multipartite system is trivial).
Each subsystem will be equipped with an algebra of observables, $\cal{A}$ and $\cal{B}$,
and the total algebra will be given by $\cal{A} \otimes \cal{B}$.
{\it Local} observables are those elements of the total algebra that can be expressed in the form $A \otimes {\mathbb I}_B$
(or ${\mathbb I}_A \otimes B$), i.e.~the observables only acting on one subsystem.
Entanglement is generally defined as quantum correlation between non-local observables,
but it can also be expressed in terms of correlation between systems, as we shall see immediately.

The state, pure or mixed, of each subsystem is described by its density matrix $\hat\rho^{(A)}$ ($\hat\rho^{(B)}$).
The state of the total system is the tensor product of the two subsystem states \emph{if and only if} there is neither statistical (classical)
nor quantum correlation:
$$
\hat\rho^{(AB)} = \hat\rho^{(A)} \otimes \hat\rho^{(B)} + \mbox {cl-corr.} + \mbox{q-corr.}
$$
If (and only if) non local q-correlation is absent, the state of the total system can be written as a mixture of uncorrelated states,
and the state is separable. We can then define:
\begin{definition}[Entanglement and separability]
A state $\hat\rho^{(AB)}$ of a bipartite system $S_{AB} = S_A \otimes S_B$ is \emph{separable}
if and only if it can be written in the form:
\begin{equation} \label{ent}
\hat\rho^{(AB)} = \sum_n w_n \hat\rho_n^{(A)} \otimes \hat\rho_n^{(B)},
\end{equation}
where the $w_n$ are classical probabilities, i.e. $w_n \geq 0 \, \forall n,
\sum_n w_n =1$. Otherwise, when the state cannot be written in this form, the
subsystems are said to be \emph{entangled}.~\footnote{See in particular Ref. \cite{diosi} sect. 2.5 and 4.5.}
\end{definition}
This definition can be easily generalized to multipartite systems.
Entanglement corresponds to q-correlation between non-local systems.

%%%%%%%%%%%%%%%%%%%
\section{Pure states as entanglement sources}

In the previous section we have defined entanglement as a property for a quantum system connected with the notion of non-locality.
Locality is only defined with respect to a TPS within the space that represents the system. In turn, a TPS is induced by the algebra 
of observables, so our next goal will be to provide sufficient criteria for the algebra of a finite-dimensional system to induce a TPS.

As an example, consider a generic pure state $|\psi \rangle \in {\cal H}$,
where ${\cal H}$ is an \emph{unstructured} $d$-dimensional H-space,
i.e.~a H-space with no TPS; we want to find a partitioning of ${\cal H}$ such that $|\psi \rangle$
has a certain desired amount of entanglement. We shall exploit the fact that all same-dimensional
H-spaces are isomorphic.~\footnote{In particular, since any $d$-dimensional H-space is isomorphic to
${\mathbb C}^d$, we will often refer to ${\mathbb C}^d$ in the rest of the chapter as to a generic
$d$-dimensional H-space, for simplicity.} Let us consider a pure state $|\phi \rangle$ in a \emph{model}
H-space, i.e.~another $d$-dimensional H-space equipped with a TPS according to our liking,
say ${\cal H}' = \bigotimes_i {\cal H}'_i$. The subsystems ${\cal H}'_i$ will have their own algebras of
observables ${\cal A}'_i$. Since ${\cal H} \cong {\cal H}'$, it is always possible to find an unitary map
$U : {\cal H} \rightarrow {\cal H}'$ such that $U |\psi \rangle = |\phi \rangle$, and $U$ will map
the ${\cal A}'_i$ back to algebras ${\cal A}_i = U {\cal A}'_i U^{\dagger}$ acting on ${\cal H}$.
We shall investigate the conditions on the ${\cal A}_i$ to induce a TPS on the previously unstructured H-space.

Our final purpose, concerning this chapter, is to demonstrate the \emph{Tailored observables theorem}
({\bf Theorem \ref{tailor}}, in {\bf Sect. \ref{stailor}}): observables can be constructed such that any pure state has any
amount of entanglement possible.

From now on we shall use some elementary algebraic definitions and properties, which are exposed in {\bf Appendix \ref{algebra}}.

\section{Induced tensor product structures}

Every subsystem of a total system is equipped with its own (local) algebra of observables; with the next
theorem we shall provide, for the case of finite-dimensional bipartite systems, sufficient criteria for those algebras to
induce a tensor product structure for the total system.

\begin{theorem}[Induced TPSs] \label{TPS}

Consider the full matrix algebra ${\cal M}_d$~\footnote{We will always denote with ${\cal M}_d$
the full matrix algebra on the finite-dimensional H-space ${\mathbb C}^d$,
i.e. ${\cal M}_d = \{(a_{ij})^d_{i,j = 1} | (\forall i, j \in \{ 1, \dots , d \}) (a_{ij} \in {\mathbb C})\}$.}
on the finite dimensional Hilbert space ${\cal H} = {\mathbb C}^d$
and two subalgebras ${\cal A}$ and ${\cal B}$ of ${\cal M}_d$, for which hold:
\begin{description}
 \item[(i)] Independence: $[{\cal A},{\cal B}] = {0}$, i.e. $[a,b] = 0 \; \forall \; (a \in {\cal A}, b \in {\cal B})$;
 \item[(ii)] Completeness: ${\cal A} \otimes {\cal B} \cong {\cal A} \vee {\cal B} = {\cal M}_d$.
\end{description}
Then ${\cal A}$ and ${\cal B}$ induce a TPS on ${\mathbb C}^d$, i.e., there exist two H-spaces ${\mathbb C}^k$ and ${\mathbb C}^l$,
$d = k \cdot l$, and an unitary mapping $U : {\mathbb C}^k \otimes {\mathbb C}^l \rightarrow {\mathbb C}^d$ such that
$$
{\cal A} = U({\cal M}_k \otimes {\mathbb I}_l) U^\dagger \mbox{\emph{ and }} {\cal B} = U({\mathbb I}_k \otimes {\cal M}_l) U^\dagger.
$$
In particular, ${\cal A}$ and ${\cal B}$ are isomorphic to ${\cal M}_k$ and ${\cal M}_l$ respectively.

\end{theorem}

\emph{About the requirements}. Requirement {\bf (i)} concerns \emph{separability}, guaranteeing the possibility to perform
measurements on one subsystem without affecting the other. Requirement {\bf (ii)} ensures that the TPS of the algebra induces
a corresponding TPS on the total state space. Another necessary requirement, not affecting the mathematical structure but fundamental
from a physical point of view, is local accessibility of the observables.

\begin{proof}
See {\bf Appendix \ref{proof}}.
\end{proof}

The generalization of this statement to multipartite systems, known under the name of \emph{Zanardi's Theorem},
is quite simply achieved by induction from the bipartite case \cite[Cor. 3]{harshman}.

\section{The \emph{Tailored observables theorem}} \label{stailor}

We have just seen under which conditions the algebras can induce a TPS on the total system.
Our last step will be to demonstrate that, for any pure state over a (finite-dimensional) H-space,
there always exist two subalgebras inducing a TPS such that the state has any possible amount of entanglement.
In order to do so, we need the following

\begin{lemma} \label{su2}

The (associative) algebra generated by any representation of $\mathfrak{su}(2)$ on ${\mathbb C}^d$, with $d \in \mathbb{N}
\setminus \{0,1\}$~\footnote{The so called spin representation: if we define $d \equiv 2s+1$, $s \in \{ \frac{1}{2}, 1, \frac{3}{2}, \cdots \}$,
$d$ is the number of spin eigenvalues $-s, -s+1, \cdots, s$.} is the full matrix algebra ${\cal M}_d$ 
\emph{\cite[Th. 5]{harshman}}.

\end{lemma}

We can now demonstrate the already mentioned \emph{Tailored observables theorem}, which also provides a constructive method, virtually
implementable in a finite number of steps, to construct the factorization that induces the necessary notion of locality,
in order to create any possible amount of entanglement for a given (pure) state.

\begin{theorem}[Tailored observables] \label{tailor}

Let ${\cal H} = {\mathbb C}^d$ be a H-space with an orthonormal basis $\{ |i \rangle \}$, and take
$k, l \in {\mathbb N}$ such that $d = k \cdot l.$
Then, for every (pure) state $| \Psi \rangle = \sum^d_{i=1} c_i|i \rangle$ and every $\lambda_1, \cdots ,\lambda_{min \{k,l\}} \in {\mathbb C}$
with $\sum^d_{i=1} |c_i|^2 = \sum^{min \{k,l\}}_{i=1} |\lambda_i|^2$, there exist algebras ${\cal A}$ and ${\cal B}$ satisfying the conditions
of {\bf Theorem \ref{TPS}}, and an unitary operator $U$ such that
$$
| \Psi \rangle = U \sum^{min \{k,l\}}_{i=1} \lambda_i |i \rangle_A |i \rangle_B,
$$
with $\{|i \rangle_A \}$ and $\{|i \rangle_B \}$ orthonormal bases for H-spaces
${\cal H}_A = {\mathbb C}^k$ and ${\cal H}_B = {\mathbb C}^l$, respectively.

\end{theorem}

\emph{Note}. The condition for $| \Psi \rangle$ to be expressible as a wavefunction implies its purity.
This hypothesis is central for the results exposed in this chapter, and will not be assumed in the rest of our work.

\begin{proof}
See {\bf Appendix \ref{proof}}.
\end{proof}

As for {\bf Theorem \ref{TPS}}, the generalization of this last result to any multipartite case is obtained
from the bipartite one by iteration, except we cannot rely on the Schmidt decomposition \cite[Cor. 7]{harshman}.

At least in principle,
for finite-dimensional systems all entanglement properties are still valid, even if, for higly partitioned systems,
it can get more and more complicated, from the computational point of view, to pratically build the desired subalgebras.

In conclusion, we have seen how it is always possilble to find (or provide) a factorization
of the algebra of observables such that any pure state has any admittable amount of entanglement.
In the next chapter we shall go further, considering the case of mixed states, which is the main theme of this work.

%%%%%%%%%%%%%%%%%
\section{Example}

Let us consider the simplest case of a H-space ${\cal H} \equiv {\mathbb C}^4$ which is \emph{unstructured}, i.e.~it has no TPS, with an
orthonormal basis $\{|0 \rangle, |1 \rangle, |2 \rangle, |3 \rangle \}$, and a pure state $|\psi \rangle \equiv |0 \rangle$.
Let us also consider a model H-space ${\cal H}' = {\cal H}_{A'} \otimes {\cal H}_{B'} \equiv {\mathbb C}^2 \otimes {\mathbb C}^2$.
Given any couple of bases $\{ |0\rangle_{A'},|1\rangle_{A'} \}$ for ${\cal H}_{A'}$ and $\{ |0\rangle_{B'},|1\rangle_{B'} \}$ for ${\cal H}_{B'}$,
we can build a basis $\{ |0\rangle',|1\rangle',|2\rangle',|3\rangle' \}$ for ${\cal H}'$ by setting, for example:
$$
|2i+j \rangle' \equiv |i\rangle_{A'} \otimes |j\rangle_{B'}.
$$

We now want to tailor algebras of observables (over ${\cal H}$) ${\cal A}$ and ${\cal B}$ that induce a partitioning such that
$|\psi \rangle$ has the same entanglement as a certain (arbitrary) $|\phi \rangle \in {\cal H}'$,
say $|\phi \rangle = \lambda_1 |0\rangle' + \lambda_2 |3\rangle'$, with $\lambda_1, \lambda_2 \in {\mathbb C}$,
and $|\lambda_1|^2 + |\lambda_2|^2 = 1$ as $|\phi \rangle$ is normalized. 
In order to do so, we shall choose an unitary operator $U$ that maps $| \psi \rangle$ to $| \phi \rangle = U | \psi \rangle$.
A simple (maybe the simplest) choice can be:
$$
U = \left(
\begin{array}{cccc}
\lambda_1 &0 &0 &\lambda_2 \\
0 &1 &0 &0 \\
0 &0 &1 &0 \\
-\lambda_2 &0 &0 &\lambda_1 \\
\end{array} \right),
$$ 
but others are possible; this freedom in the choice of $U$ could be exploited if there were additional
practical constraints on the kind of measurements to take into account.

At this point we are able to use the new operator to map back ${\cal A}'$ and ${\cal B}'$ to ${\cal A} $ and ${\cal B}$:
$$
{\cal A} = U^{\dagger} {\cal A}' U, \qquad {\cal B} = U^{\dagger} {\cal B}' U. 
$$
Thanks to {\bf Lemma \ref{su2}} we know that ${\cal A}'$ can be generated by the operators:~\footnote{$\sigma_i$ are the well-known Pauli matrices,
see {\bf Sect. \ref{pauli}}.}

\[
\begin{split}
S^{A'}_x &= \frac{1}{2} \sigma_x \otimes {\mathbb I}_2 = \frac{1}{2} \left(
\begin{array}{cccc}
 0 & 0 & 1 & 0 \\
 0 & 0 & 0 & 1 \\
 1 & 0 & 0 & 0 \\
 0 & 1 & 0 & 0 \\
\end{array} \right), \\
S^{A'}_y &= \frac{1}{2} \sigma_y \otimes {\mathbb I}_2 = \frac{1}{2} \left(
\begin{array}{cccc}
 0 & 0 & -i & 0 \\
 0 & 0 & 0 & -i \\
 -i & 0 & 0 & 0 \\
 0 & -i & 0 & 0 \\
\end{array} \right), \\
S^{A'}_z &= \frac{1}{2} \sigma_z \otimes {\mathbb I}_2 = \frac{1}{2} \left(
\begin{array}{cccc}
 1 & 0 & 0 & 0 \\
 0 & 1 & 0 & 0 \\
 0 & 0 & -1 & 0 \\
 0 & 0 & 0 & -1 \\
\end{array} \right).
\end{split}
\]
Similarly ${\cal B}'$ can be generated by:

\[
\begin{split}
S^{B'}_x &= \frac{1}{2} {\mathbb I}_2 \otimes \sigma_x = \frac{1}{2} \left(
\begin{array}{cccc}
 0 & 1 & 0 & 0 \\
 1 & 0 & 0 & 0 \\
 0 & 0 & 0 & 1 \\
 0 & 0 & 1 & 0 \\
\end{array} \right), \\
S^{B'}_y &= \frac{1}{2} {\mathbb I}_2 \otimes \sigma_y = \frac{1}{2} \left(
\begin{array}{cccc}
 0 & -i & 0 & 0 \\
 i & 0 & 0 & 0 \\
 0 & 0 & 0 & -i \\
 0 & 0 & i & 0 \\
\end{array} \right), \\
S^{B'}_z &= \frac{1}{2} {\mathbb I}_2 \otimes \sigma_z = \frac{1}{2} \left(
\begin{array}{cccc}
 1 & 0 & 0 & 0 \\
 0 & -1 & 0 & 0 \\
 0 & 0 & 1 & 0 \\
 0 & 0 & 0 & -1 \\
\end{array} \right).
\end{split}
\]
So we can easily calculate the generators of the tailored subalgebras. For ${\cal A}$ we have:

\[
\begin{split}
S^{A}_x &= U^{\dagger} S^{A'}_x U = \frac{\lambda_1}{2} \sigma_x \otimes {\mathbb I}_2 + \frac{\lambda_2}{2} \sigma_z \otimes \sigma_x \\
&= \frac{1}{2} \left(
\begin{array}{cccc}
 0 & \lambda_2 & \lambda_1 & 0 \\
 \lambda_2 & 0 & 0 & \lambda_1 \\
 \lambda_1 & 0 & 0 & -\lambda_2 \\
 0 & \lambda_1 & -\lambda_2 & 0 \\
\end{array} \right), \\
S^{A}_y &= U^{\dagger} S^{A'}_y U = \frac{\lambda_1}{2} \sigma_y \otimes {\mathbb I}_2 - \frac{\lambda_2}{2} \sigma_z \otimes \sigma_y \\
&= \frac{1}{2} \left(
\begin{array}{cccc}
 0 & i\lambda_2 & -i\lambda_1 & 0 \\
 -i\lambda_2 & 0 & 0 & -i\lambda_1 \\
 i\lambda_1 & 0 & 0 & -i\lambda_2 \\
 0 & i\lambda_1 & i\lambda_2 & 0 \\
\end{array} \right), \\
  S^{A}_z &= U^{\dagger} S^{A'}_z U \\
&= \frac{1}{2} [\lambda^2_1 \sigma_z \otimes {\mathbb I}_2 - \lambda^2_2 {\mathbb I}_2 \otimes \sigma_z
- \lambda_1 \lambda_2 \sigma_x \otimes \sigma_x + \lambda_1 \lambda_2 \sigma_y \otimes \sigma_y ] \\
&= \frac{1}{2} \left(
\begin{array}{cccc}
 \lambda^2_1 - \lambda^2_2 & 0 & 0 & -2 \lambda_1 \lambda_2 \\
 0 & 1 & 0 & 0 \\
 0 & 0 & -1 & 0 \\
 -2 \lambda_1 \lambda_2 & 0 & 0 & -\lambda^2_1 + \lambda^2_2 \\
\end{array} \right). \\
\end{split}
\]
We can obtain the generators of ${\cal B}$ in the same way, transposing the order of the Pauli matrices.

We have then tailored the desired subalgebras over the unstructured H-space ${\cal H}$. It is easy to verify that ${\cal A}$ and ${\cal B}$
satisfy the conditions of {\bf Theorem \ref{TPS}} and can always be chosen so that any pure state $| \psi \rangle$
has any amount of entanglement desired, from the situation where there is no entanglement ($\lambda_1 = 0 \, \vee \, \lambda_2 = 0$)
to the one where entanglement is maximal ($\lambda_1 = \lambda_2 = \frac{1}{\sqrt{2}}$).

%%%%%%%%%%%%%%%%%%
\chapter{Entanglement in mixed states} \label{chap3}

In this chapter, we treat the case of mixed states, elaborating the recent exposition in \cite{thirring}.
Here the situation is more complex than for pure states: 

\begin{itemize}
    \item concerning {\bf separability}, nothing changes, i.e., it is always possible to find a factorization such that any given mixed state is separable;
    \item concerning {\bf entanglement}, the statement demonstrated in the previous chapter does not hold in general for mixed states;
for example, it is immediate to see that the \emph{tracial state} $\hat\rho = \frac{1}{d} {\mathbb I}_d$ is separable for any
factorization.~\footnote{See eq. (\ref{tracial}).}
\end{itemize}

From now on we shall work with a bipartite H-space ${\cal H} = {\cal A} \otimes {\cal B}$
with dimension $D = d_1 \times d_2$; for simplicity, we will often set $d_1 = d_2 \equiv d$.
Some of the properties we are going to discuss can be extended to generic multipartite systems,
but this is beyond the scope of our present work.

Mixed states ore only expressible in terms of density matrices, elements of the algebra of observables which,
as far as we consider finite dimensions, is an algebra of hermitian matrices
equipped with the scalar product $\langle A| B \rangle \equiv \mbox{Tr}(A^\dagger B) = \mbox{Tr}(A B)$
and the corresponding norm (\emph{Frobenius norm}, or \emph{Hilbert-Schmidt norm})
$\lVert A \rVert \equiv \sqrt{\langle A| A \rangle}$.
Such a norm can be used to define a \emph{Hilbert-Schmidt distance} for the algebra:
\begin{equation} \label{dist}
\delta (A,B) \equiv \lVert A - B \rVert = \sqrt{\mbox{Tr}[(A - B)^2]}.
\end{equation}
Of course, other norms and distances could be defined, but this is not our purpose here.

\section{Separability}

The following theorem states that, for any state, it is always possible to find a factorization
such that it appears separable.

\begin{theorem}[Separability for mixed states]
Given any state $\hat\rho$ in a $D$-dimensional H-space, and any factorization ${\cal M}$ of the algebra of observables,
it is always possible to find an unitary operator $U$ such that $\hat\rho$ appears separable with respect to $U^\dagger \! {\cal M} U$.
\end{theorem}

\begin{proof}
Consider a factorization ${\cal A} \otimes {\cal B} = {\cal H}$ with dimension $d_1 \times d_2 = D$,
and the corresponding partition ${\cal A}' \otimes {\cal B}'$ of the algebra.
We can find three bases $| \psi_i \rangle$ on ${\cal H}$, $| \psi^{(1)}_j \rangle$
on ${\cal A}$, $| \psi^{(2)}_k \rangle$ on ${\cal B}$
($i = 1, \cdots , D$; $j = 1, \cdots , d_1$; $k = 1, \cdots , d_2$). Any state can be expanded in the form
$$
\hat\rho = \sum^D_{i=1} \rho_i | \psi_i \rangle \langle \psi_i |,
$$
where the $\rho_i$ are scalars.

We know that the set $\{ | \psi^{(1)}_j \rangle \otimes | \psi^{(2)}_k \rangle \}_{j,k=1}^{d_1,d_2}$ is a basis for ${\cal H}$.
We can order the $D$ couples $(j,k)$ according to our liking and identify each $i$ with a couple through a
bijective map: $i \equiv i (j,k)$. There will thus exists an unitary operator $U$ such that:
$$
U | \psi_{i(j,k)} \rangle = | \psi^{(1)}_j \rangle \otimes | \psi^{(2)}_k \rangle, \, (i = 1, \cdots , D). 
$$
We can then write:
$$
\hat\zeta \equiv U \hat\rho \, U^\dagger = \sum_{i=1}^D \rho_{i(j,k)} | \psi^{(1)}_j \rangle \langle \psi^{(1)}_j |
\otimes | \psi^{(2)}_k \rangle \langle \psi^{(2)}_k |,
$$
and $\hat\zeta$ is clearly separable, according to definition (\ref{ent}). This means $\hat\rho$ is separable
with respect to the factorization algebra $U^\dagger \! ({\cal A}' \otimes {\cal B}') U$.~\footnote{This actually is a 
factorization algebra: see {\bf Theo. \ref{tailor}}.}
\end{proof}

A similar procedure cannot be applied to find a factorization where $\hat\rho$ is entangled because, unlike the set of
separable states, the set of entangled states is not convex, i.e., a linear hull of entangled states is not
ncessarily an entangled state itself. We could expand the density matrix over a set of (even maximally) entangled states,
but we could not go further \cite[Sect. II]{thirring}.

%%%%%%%%%%%%%%%%%%%%%%
\section{Entanglement}

A general entanglement criterion, which would allow us to check whether any composite system has some quantum correlation within,
has not been discovered yet. We are thus going to show two particular cases for which a sufficient condition has been found.
In order to do so we shall introduce some tools first.

%%%%%%%%%%%%%%%%%%%%%%%%%%%%%%%%
\subsection{Entanglement witness} \label{ew0}

We have given a precise definition of separability and entanglement in the beginning
of {\bf Chap. \ref{chap2}}, but let us show an ideal method to verify whether a certain state is entangled or not:

\begin{theorem}[Entanglement witness]
 If $\hat\rho_{ent}$ is entangled, then there exists a hermitian operator $E$ such that:
\begin{equation} \label{ew1}
\langle \hat\rho_{ent} | E \rangle = Tr(\hat\rho_{ent}E) < 0.
\end{equation}
Furthermore, for all separable $\hat\rho_{sep}$:
\begin{equation} \label{ew2}
\langle \hat\rho_{sep} | E \rangle = Tr(\hat\rho_{sep} E) \geq 0.
\end{equation}
$E$ is called an \emph{entanglement witness} (EW) \emph{\cite[Th. 5]{harshman}}.
\end{theorem}

Let us also define:

\begin{definition}[Optimal EW]
 An entanglement witness is \emph{optimal} if there exists a separable $\hat\rho_{sep}$ such that
\begin{equation} \label{ew3}
\langle \hat\rho_{sep} | E_{opt} \rangle = 0.
\end{equation}
\end{definition}

Given an entangled state, it is always possible to find an EW, but the search may certainly result in an enormous
computational effort, especially for large dimensions or multiple factorizations.
Operationally, we can build an  optimal EW as follows \cite[(39)]{thirring}:
\begin{equation} \label{opt}
E_{opt} \equiv \frac{\hat\rho_0 - \hat\rho_{ent} - \langle \hat\rho_0 | \hat\rho_0 - \hat\rho_{ent} \rangle {\mathbb I}}
                    {\lVert \hat\rho_0 - \hat\rho_{ent} \rVert},
\end{equation}
where $\hat\rho_0$ is the nearest to $\hat\rho_{ent}$ among all the separable states, but again the determination
of $\hat\rho_0$ may not be really easy.~\footnote{We prove that $E_{opt}$ is an optimal EW in {\bf Sect. \ref{ew}}.}

%%%%%%%%%%%%%%%%%%%%%%%%%%%%%%%%%%%%%%
\subsection{Peres-Horodecki criterion} \label{pauli}

We shall now introduce the Peres-Horodecki criterion, which poses a condition that is
necessary for separability, and consequently sufficient for entanglement. 
Furthermore, in low dimension this condition is equivalent to separability.
Let us begin with an illustrative example.

Consider a qubit. It is easily demonstrated that every density matrix can
be written in the form \cite[Sect. 5.2.3]{diosi}:
$$
\hat\rho = \frac{{\mathbb I} + \vec{s} \cdot \vec{\sigma}}{2}, \quad |\vec{s}| \leq 1,
$$
where $\vec{s}$ is a three-dimensional vector and $\vec{\sigma}$ is the \emph{spin vector},
whose components are the three \emph{Pauli matrices}:
\[
\sigma_x =
\left(
\begin{array}{cc}
 0 & 1 \\
 1 & 0 \\
\end{array} \right),
\sigma_y =
\left(
\begin{array}{cc}
 0 & -i \\
 i & 0 \\
\end{array} \right),
\sigma_z =
\left(
\begin{array}{cc}
 1 & 0 \\
 0 & -1 \\
\end{array} \right),
\]
$|\vec{s}| = 1$ if and only if the state is pure, i.e., it exists a $| \psi \rangle$ such that:
$$
| \psi \rangle \langle \psi | = \frac{{\mathbb I} + \vec{s} \cdot \vec{\sigma}}{2}.
$$
A hypothetical spin-inversion operator $T$ would act on a state this way:
$$
T \frac{{\mathbb I} + \vec{s} \cdot \vec{\sigma}}{2} = \frac{{\mathbb I} - \vec{s} \cdot \vec{\sigma}}{2}.
$$

Consider now a maximally entangled state for two qubits
$$
| \psi \rangle = \frac{{\mathbb I} \otimes {\mathbb I} - \vec{\sigma} \otimes \vec{\sigma}}{4};
$$
a partial, or local, spin-inversion $T \otimes {\mathbb I}$ would transform $| \psi \rangle$ into:
$$
(T \otimes {\mathbb I}) \frac{{\mathbb I} \otimes {\mathbb I} - \vec{\sigma} \otimes \vec{\sigma}}{4} =
\frac{{\mathbb I} \otimes {\mathbb I} + \vec{\sigma} \otimes \vec{\sigma}}{4},
$$
but this last matrix is no longer a density matrix, being indefinite, as we can see for instance from:
$$
\mbox{Tr} \left( \frac{{\mathbb I} \otimes {\mathbb I} - \vec{\sigma} \otimes \vec{\sigma}}{4} \;	
\frac{{\mathbb I} \otimes {\mathbb I} + \vec{\sigma} \otimes \vec{\sigma}}{4} \right) = - \frac{1}{2}.
$$

It could be shown that the partial spin-inversion transforms every entangled matrix into an indefinite matrix.
Physically, this phenomenon occurs when a system is in an entangled state, since an inversion
of the state of a single subsystem is not actually realizable without affecting the
rest of the system.

This property is formalized for the most general case in the following very powerful criterion,
necessary for separability:

\begin{theorem}[Peres-Horodecki criterion] \label{phcrit}
Consider the state $\hat\rho$ for a bipartite system ${\cal A} \otimes {\cal B}$,
where $\hat\rho^{(A)}$ denotes the partial transpose of $\hat\rho$ with respect to the subsystem ${\cal A}$.
If $\hat\rho$ is separable, then $\hat\rho^{(A)} \geq 0$ \emph{\cite[Th. 1]{lewenstein}}.
\end{theorem}
A matrix that fulfills the above condition is called ``PPT'', for \emph{positive partial transpose}.

Since $\left( \hat\rho^{(A)} \right)^{(B)} = \hat\rho^T \geq 0$, then $\hat\rho^{(A)} \geq 0$ implies
$\hat\rho^{(B)} \geq 0$ and vice versa, as we expected, since the two subsystems are totally equivalent
for what concerns the criterion.

In $2 \times 2$ or $2 \times 3$ dimensions the condition is also sufficient:

\begin{theorem}
 In spaces of dimensions $2 \times 2$ or $2 \times 3$, $\hat\rho$ is separable if and
only if $\hat\rho^{(A)} \geq 0$.
 \emph{\cite[Th. 2]{lewenstein}}.
\end{theorem}

%%%%%%%%%%%%%%%%%%%%%%%%%%%
\subsection{Splitted states}

We shall now introduce the first class of states which can be represented as entangled.
We will provide and use an optimal EW in the proof.~\footnote{For the definition of
optimal EW, see {\bf Subs. \ref{ew0}}.}

\begin{theorem}[Bound for splitted states] \label{splitteo}
Consider a state $\hat\rho$ over a $d \times d$-dimensional H-space.
Let $\hat\pi$ be a projector to a maximally entangled pure state, $\hat\sigma$ be another state, orthogonal
to $\hat\pi$, i.e. $\langle \hat\pi|\hat\sigma \rangle = 0$ and $\beta$ be a real number, with $0 \leq \beta \leq 1$.
If $\hat\rho$ can be written in the splitted form
\begin{equation} \label{split}
\hat\rho = \beta \hat\pi + (1 - \beta) \hat\sigma,
\end{equation}
then $\hat\rho$ is entangled if $\beta > \frac{1}{d}$. 
\end{theorem}

\begin{proof}
 Consider the operator $A = {\mathbb I}_{d^2} - d \hat\pi$. We will show that $A$ is an optimal entanglement witness.
First, we show that for any pure $\hat\rho_{sep} \equiv |\phi \otimes \psi \rangle \langle \phi \otimes \psi |$
it holds $\langle \hat\rho_{sep}| A \rangle = \mbox{Tr} \hat\rho_{sep} A \geq 0$, see definition (\ref{ew2}). The extension to mixed states
is achieved by linearity. Since $\hat\pi$ is a projector to a maximally entangled state, we can define 
$\hat\pi \equiv | \chi \rangle \langle \chi |$ with $| \chi \rangle  \equiv
\frac{1}{\sqrt{d}} \sum_i | \chi_i^{(A)} \rangle \otimes  |\chi_i^{(B)} \rangle$. Both
$\{|\chi_i^{(A)} \rangle\}_{i=1}^d$ and $\{|\chi_i^{(B)} \rangle\}_{i=1}^d$ are a basis for their
respective subspace; so:
\[
 \begin{split}
  \langle \phi \otimes \psi | A | \phi \otimes \psi \rangle &= 1 - d \langle \phi \otimes \psi | \hat\pi | \phi \otimes \psi \rangle \\
&= 1- d \sum_{i,j=1}^d \frac{1}{d} \langle \phi | \chi_i^{(A)} \rangle \langle \psi | \chi_i^{(B)} \rangle
\langle \chi_i^{(A)} | \phi \rangle \langle \chi_i^{(B)} | \psi \rangle \\
&\equiv 1 - d \sum_{i,j=1}^d \frac{1}{d} \phi_i^* \psi_i^* \phi_j \psi_j \\
&= 1 - | \langle \phi^* | \psi \rangle |^2 \geq 0.
 \end{split}
\]

For $| \phi^* \rangle = | \psi \rangle $, we have $\langle \phi \otimes \psi | A | \phi \otimes \psi \rangle = 0$,
so the witness is optimal by (\ref{ew3}).

The scalar product between $\hat\rho$ and $A$ is:
\[
 \begin{split}
  \langle \hat\rho | A \rangle 
  &= \langle \beta \hat\pi + (1 - \beta) \hat\sigma | {\mathbb I}_{d^2} - d \hat\pi \rangle  \\
  &= \beta \langle \hat\pi | {\mathbb I}_{d^2} \rangle + (1 - \beta) \langle \hat\sigma | {\mathbb I}_{d^2} \rangle
  - d \beta \langle \hat\pi | \hat\pi \rangle - d (1 - \beta) \langle \hat\sigma | \hat\pi \rangle \\
&= \beta + (1 - \beta) - d \beta =  1 - d \beta,
 \end{split}
\]
where we have used the property of the density matrix $\langle \hat\rho | {\mathbb I} \rangle
= \mbox{Tr} (\hat\rho {\mathbb I}) = 1$ (see {\bf Subs. \ref{density}}), and the identity $\hat\pi = \hat\pi^2$,
due to the fact that $\hat\pi$ is a projector.

Then we have:
$$
\langle \hat\rho | A \rangle < 0 \quad \Longleftrightarrow \quad \beta > \frac{1}{d},
$$
which was to be demonstrated.
\end{proof}

Note that every state which can be decomposed in the form (\ref{split}) has maximum eigenvalue $\rho_{max} = \max \{\beta, (1-\beta)\}$,
so, to satisfy the hypothesis, it must hold in particular $\rho_{max} > \frac{1}{d}$.

%%%%%%%%%%%%%%%%%%%%%%%%%%%%%%%%%%%%%%%%%
\subsection{Entanglement over a subspace}

We will now show how, under certain conditions, it is possible to create entanglement
over the subspace generated by the eigenvectors of a density matrix. 

\begin{theorem}[Entanglement over a subspace]
 
Consider a state $\hat\rho$ in a space factorized with dimension $d \times d$.
Let its eigenvalues be $\{\rho_i\}_{i=1}^{d^2}$, ordered in the following way:
$$
1 \geq \rho_1 \geq \rho_2 \geq \cdots \geq \rho_{d^2-1} \geq \rho_{d^2} \geq 0,
$$
where the bounds are due to the properties of the density operators. We will
also call $| i \rangle$ the eigenvector of $\rho_i$.  Consider the subspace
$$
{\cal V} \equiv \mbox{\emph{span}}(|1\rangle, |d^2-2\rangle, |d^2-1\rangle, |d^2\rangle).
$$
If $\rho_1 > \frac{3}{d^2}$, then there is always a choice of factorization possible
such that $\hat\rho$ appears entangled in ${\cal V}$.

\end{theorem}

\begin{proof}
 Over the basis $\{ | i \rangle \}$, the matrix $\hat\rho$ is diagonal:
\[
 \left(
\begin{array}{ccccc}
\ddots & \vdots & \vdots & \vdots & \vdots \\
\cdots & \rho_1 & 0 & 0 & 0 \\
\cdots & 0 & \rho_{d^2-2} & 0 & 0 \\
\cdots & 0 & 0 & \rho_{d^2} & 0 \\
\cdots & 0 & 0 & 0 & \rho_{d^2-1} \\
\end{array}
\right);
\]
we have reordered the basis vectors to simplify our next steps.

The four-dimensional space ${\cal V}$ can be factorized as a $2 \times 2$ dimensional tensor product space
${\cal W}_a \otimes {\cal W}_b = {\cal W}$, and in particular we choose the vectors $|1\rangle$ and $|d^2-1\rangle$
to appear maximally entangled in the new factorization.~\footnote{This is always possible, as we have proved
in {\bf Chap. \ref{chap2}}.} We define $K: {\cal V} \longrightarrow {\cal W}$,
 the unitary operator connecting the two factorizations, in the following way:
\begin{align*}
 K |1\rangle &\equiv \frac{|\uparrow_a \uparrow_b\rangle + |\downarrow_a  \downarrow_b \rangle}{\sqrt{2}}; \\
 K |d^2-2\rangle &\equiv |\uparrow_a \downarrow_b \rangle; \\
 K |d^2\rangle &\equiv |\downarrow_a \uparrow_b \rangle; \\
 K |d^2-1\rangle &\equiv \frac{|\uparrow_a \uparrow_b\rangle - |\downarrow_a \downarrow_b \rangle}{\sqrt{2}};
\end{align*}
$\{ |\uparrow_a\rangle, |\downarrow_a\rangle \}$ and $\{ |\uparrow_b\rangle, |\downarrow_b\rangle \}$ are bases
of the two subspaces.

Represented with respect to the bases $\{ |1\rangle, |d^2-2\rangle, |d^2\rangle, |d^2-1\rangle \}$ of ${\cal V}$
and $\{ |\uparrow_a \uparrow_b \rangle, |\uparrow_a \downarrow_b \rangle,
|\downarrow_a \uparrow_b \rangle, |\downarrow_a \downarrow_b \rangle \}$ of ${\cal W}$, $K$ assumes the
form:
\[
 \left(
\begin{array}{cccc}
 \frac{1}{\sqrt{2}} & 0 & 0 & \frac{1}{\sqrt{2}} \\
 0 & 1 & 0 & 0 \\
 0 & 0 & 1 & 0 \\
 \frac{1}{\sqrt{2}} & 0 & 0 & -\frac{1}{\sqrt{2}} \\
\end{array}
\right),
\]
so it is easy to see that it can be written:~\footnote{We consider the restriction of $\hat\rho$ to $V$,
the rest of the matrix is unaffected by the transformation.}
\[
 \hat\rho_K \equiv K \hat\rho K^\dagger =
\left(
\begin{array}{cccc}
 \frac{1}{2}(\rho_1+\rho_{d^2-1}) & 0 & 0 & \frac{1}{2}(\rho_1-\rho_{d^2-1}) \\
 0 & \rho_{d^2-2} & 0 & 0 \\
 0 & 0 & \rho_{d^2} & 0 \\
 \frac{1}{2}(\rho_1-\rho_{d^2-1}) & 0 & 0 & \frac{1}{2}(\rho_1+\rho_{d^2-1}) \\
\end{array}
\right).
\]

Our scope, now, is to show that $\hat\rho_K$ violates the Peres-Horodecki criterion ({\bf Theo. \ref{phcrit}}), i.e., it
is entangled. Consider the partial transpose of $\hat\rho_K$ with respect to the subspace ${\cal W}_b$:
\[
 \hat\rho_K^{(b)} =
\left(
\begin{array}{cccc}
 \frac{1}{2}(\rho_1+\rho_{d^2-1}) & 0 & 0 & 0 \\
 0 & \rho_{d^2-2} & \frac{1}{2}(\rho_1-\rho_{d^2-1}) & 0 \\
 0 & \frac{1}{2}(\rho_1-\rho_{d^2-1}) & \rho_{d^2} & 0 \\
 0 & 0 & 0 & \frac{1}{2}(\rho_1+\rho_{d^2-1}) \\
\end{array}
\right);
\]
we shall now search sufficient conditions for the coefficients $\rho_i$ for one of the eigenvalues of
$\hat\rho_K^{(b)}$ to be negative. 
Those eigenvalues are:
\begin{align*}
 e_1 &\equiv \frac{1}{2}(\rho_1+\rho_{d^2-1}) \mbox{, doubly degenerate}; \\
 e_2^\pm &\equiv \frac{1}{2}(\rho_{d^2-2}+\rho_{d^2}) \pm \sqrt{\frac{1}{4}(\rho_{d^2-2}+\rho_{d^2})^2
 - \rho_{d^2-2} \cdot \rho_{d^2} + \frac{1}{4}(\rho_1 - \rho_{d^2-1})^2}. \\
\end{align*}
The one and only eigenvalue which is not always non-negative is $e_2^-$, and it is negative when:
\begin{gather*}
 \frac{1}{4}(\rho_1 - \rho_{d^2-1})^2 - \rho_{d^2-2} \cdot \rho_{d^2} > 0 \\
\Longleftrightarrow \quad \frac{1}{2}(\rho_1 - \rho_{d^2-1}) > \sqrt{\rho_{d^2-2} \cdot \rho_{d^2}};
\end{gather*}
note that both terms of the first inequality are positive by assumption. Furthermore, the second term
is a geometric mean, so it is always less than or equal to the arithmetic mean of the two coefficients.
We can then relax our condition as follows:
$$
\frac{\rho_1 - \rho_{d^2-1}}{2} > \frac{\rho_{d^2-2} + \rho_{d^2}}{2} \quad
\Longleftrightarrow \quad \rho_1 > \rho_{d^2-2} + \rho_{d^2-1} + \rho_{d^2}.
$$
Since, by assumption again, it holds $\rho_{d^2-2} \geq \rho_{d^2-1} \geq \rho_{d^2}$, then the sufficient condition
for $\hat\rho_K$ to be entangled can be further relaxed to $\rho_1 > 3 \rho_{d^2-2}$.

Consider the trace of $\hat\rho$:
\[
 \mbox{Tr}\hat\rho = \sum_{i=1}^{d^2} \rho_i = 1 \quad
\Longleftrightarrow \quad 1 - \rho_1 = \sum_{i=2}^{d^2} \rho_i \geq \sum_{i=2}^{d^2-2} \rho_i.
\]
The last sum is composed of $(d^2-3)$ terms, whose minimum is $\rho_{d^2-2}$; so it must hold
\begin{gather*}
 1 - \rho_1 \geq \sum_{i=2}^{d^2-2} \rho_i \geq (d^2-3) \rho_{d^2-2} \\
\Longleftrightarrow \quad \rho_{d^2-2} \leq \frac{1 - \rho_1}{d^2-3}.
\end{gather*}
So we can finally say that a sufficient condition for $\hat\rho_K$ to be entangled
is
\[
 3 \frac{1 - \rho_1}{d^2-3} < \rho_1 \quad \Longleftrightarrow \quad \rho_1 > \frac{3}{d^2},
\]
which was to be demonstrated.
\end{proof}

This theorem allows us to find a factorization which ``splits'' the states into a maximally entangled state
(the eigenstate of $\rho_1$ or $\rho_{d^2-1}$) and a separable one (the other two eigenstates),
i.e.~to write a state in the form (\ref{split}).
We can then apply {\bf Theo.~\ref{splitteo}} to find a corresponding 
optimal EW with the same procedure.

%%%%%%%%%%%%%%%%%
\section{Absolute separability}
We have seen some of the possible constraints under which a mixed state can be represented as entangled.
In general, no final conclusions have been drawn in the literature yet. We shall now pose the opposite question:
are there (mixed) states that cannot appear entangled, for any possible factorization?
The answer is, of course, yes; the class of so called \emph{absolutely separable} states is not empty and its
``prototype'' is the normalized identity $\frac{1}{d} {\mathbb I}_d$. In fact, given any unitary $U$,
by definition,
\begin{equation}\label{tracial}
 U^\dagger \left( \frac{1}{d}  {\mathbb I}_d \right) U = \frac{1}{d} U^\dagger U = \frac{1}{d} {\mathbb I}_d.
\end{equation}
So the state remains separable, in particular, for all possible factorizations.

%%%%%%%%%%%%%%%%%%%%%%%%%%%%%%%%%%%%
\subsection{Ku\'s-\.Zyczkowski ball}

Absolute separability holds not only for the identity; as an example, any state belonging to a certain
neighbourhood of the identity has this property,~\footnote{``Neighbourhood'' here refers to a particular
notion of distance, the HS distance (\ref{dist}).} as we will state in the next theorem.
Let us first present a new definition:

\begin{definition}[Ku\'s-\.Zyczkowski ball]
Consider a bipartite $D$-dimensional system.
The \emph{Ku\'s-\.Zyczkowski ball} $B (\frac{1}{D}{\mathbb I}_D, r)$ is the maximal ball
(i.e.~the ball of maximal radius) of states of center $\frac{1}{D}{\mathbb I}_D$ which can be inscribed
into the set of separable states.
\end{definition}

The KZ ball is a set of absolutely separable states:

\begin{theorem}[Absolute separability of the Ku\'s-\.Zyczkowski ball] \label{kz}
 The radius of the KZ ball is $r_B = \frac{1}{\sqrt{D(D-1)}}$.
Every state belonging to the KZ ball is absolutely separable \emph{\cite[Th.4]{thirring}}.
\end{theorem}

%%%%%%%%%%%%%%%%%%%%%%%%%%%%%%
\subsection{Absolute separability for two qubits}
The set of absolutely separable states does not coincide with the KZ ball.
Characterization of such a set is an argument widely treated in the literature, just like, on the other hand,
the somewhat complementary characterization of the set of maximally entangled states.

The argument is discussed, for example, in \cite{verstraete} for the case of $2 \times 2$ dimensions;
here we shall just cite one main result, also reported in \cite[Lem. 2]{thirring}:
\begin{lemma}[Absolute separability for two qubits] \label{twoq}
 Let $\hat\rho$ be a state for a $2 \times 2$ dimensional system and
$\rho_1 \geq \rho_2 \geq \rho_3 \geq \rho_4$ its eigenvalues.
If it holds that
$$
\rho_1 - \rho_3 - 2 \sqrt{\rho_2 \rho_4} \leq 0,
$$
then $\hat\rho$ is absolutely separable.
\end{lemma}

%%%%%%%%%%%%%%%%%%%
\subsection{Example}

We now show some examples of absolutely and not absolutely separable states
in the two-qubit case.
Consider, over a two-qubit state ${\cal A} \otimes {\cal B}$,~\footnote{Here ${\mathbb I}$ denotes the normalized identity
$\frac{1}{4}{\mathbb I}_4$.} the density matrix
\begin{equation} \label{rhomax}
\begin{split}
\hat\rho_U = \frac{1}{4}
\left(
\begin{array}{cccc}
 1 & 0 & 0 & 1 \\
 0 & 1 & 1 & 0 \\
 0 & 1 & 1 & 0 \\
 1 & 0 & 0 & 1 \\
\end{array}
\right) \\
= \frac{1}{4} ({\mathbb I} \otimes {\mathbb I} + \sigma_x \otimes \sigma_x).
\end{split}
\end{equation}
This state is separable but it does not belong to the KZ ball, in fact
\[
 \begin{split}
  (\hat\rho_U - {\mathbb I})^2 &= \frac{1}{16}
\left(
\begin{array}{cccc}
 0 & 0 & 0 & 1 \\
 0 & 0 & 1 & 0 \\
 0 & 1 & 0 & 0 \\
 1 & 0 & 0 & 0 \\
\end{array}
\right)^2 \\ 
&= \frac{1}{16}
\left(
\begin{array}{cccc}
 1 & 0 & 0 & 0 \\
 0 & 1 & 0 & 0 \\
 0 & 0 & 1 & 0 \\
 0 & 0 & 0 & 1 \\
\end{array}
\right) \\ &= \frac{1}{4}{\mathbb I},
 \end{split}
\]
so it holds
\[
    \delta(\hat\rho_U , {\mathbb I})
    = \sqrt{\mbox{Tr}[(\hat\rho_U - {\mathbb I})^2]}
    = \frac{1}{2}\sqrt{\mbox{Tr}({\mathbb I})} = \frac{1}{2} \geq \frac{1}{\sqrt{12}} = r_B.
\]

It may then exist an unitary operator which tranforms $\hat\rho_U$ into an entangled state,
that is, for example, the case of:
\[
K = \frac{1}{\sqrt{2}} \left(
\begin{array}{cccc}
 1 & 0 & 0 & 1 \\
 0 & \sqrt{2} & 0 & 0 \\
 0 & 0 & \sqrt{2} & 0 \\
 -1 & 0 & 0 & 1 \\
\end{array}
\right)
\]

\[ \Longrightarrow
\hat\rho_{KU} \equiv K \hat\rho_U K^\dagger =
\frac{1}{4} \left(
\begin{array}{cccc}
 2 & 0 & 0 & 0 \\
 0 & 1 & 1 & 0 \\
 0 & 1 & 1 & 0 \\
 0 & 0 & 0 & 0 \\
\end{array}
\right).
\]
$\hat\rho_{KU}$ is an entangled state as it is not PPT; this we see immediately by
applying the partial transpose $(\hat\rho_{KU})^{(B)}$ to a vector $\vec{x}_U \equiv (1 - \sqrt{2}, 0, 0, 1)$:
\[
 (\hat\rho_{KU})^{(B)} \vec{x}_U = \frac{1}{4}
\left(
\begin{array}{cccc}
 2 & 0 & 0 & 1 \\
 0 & 1 & 0 & 0 \\
 0 & 0 & 1 & 0 \\
 1 & 0 & 0 & 1 \\
\end{array}
\right)
\left(
\begin{array}{c}
 1 - \sqrt{2} \\
 0 \\
 0 \\
 1 \\
\end{array}
\right) =
\frac{1 - \sqrt{2}}{4} \vec{x}_U.
\]
Correspondingly, $\hat\rho_U$ does not satisfy the assumption of {\bf Lemma \ref{twoq}},
being its spectrum $\{ \frac{1}{2}, \frac{1}{2}, 0, 0 \}$:
$$
\rho_{U1} - \rho_{U3} - 2 \sqrt{\rho_{U2} \rho_{U4}} = \frac{1}{2} > 0,
$$
as we expected, being the state not absolutely separable.

As an example of absolutely separable state we can take:
\[
 \hat\rho_V \equiv \frac{1}{4}
\left(
\begin{array}{cccc}
 1 & 0 & 0 & \frac{1}{2} \\
 0 & 1 & \frac{1}{2} & 0 \\
 0 & \frac{1}{2} & 1 & 0 \\
 \frac{1}{2} & 0 & 0 & 1 \\
\end{array}
\right),
\]
with spectrum $\{ \frac{3}{2}, \frac{3}{2}, \frac{1}{2}, \frac{1}{2} \}$.
$\hat\rho_V$ is quite ``similar'' to $\hat\rho_U$, but it satisfies both conditions of 
{\bf Lemma \ref{twoq}} and {\bf Lemma \ref{kz}}:
\begin{gather*}
    \delta(\hat\rho_V , {\mathbb I})
    = \sqrt{\mbox{Tr}[(\hat\rho_V - {\mathbb I})^2]}
    = \frac{1}{4}\sqrt{\mbox{Tr}({\mathbb I})} = \frac{1}{4} \leq \frac{1}{\sqrt{12}} = r_B, \\    
    \rho_{U1} - \rho_{U3} - 2 \sqrt{\rho_{U2} \rho_{U4}} = 1 - \sqrt{3} < 0.
\end{gather*}

It would be a long and not simple challenge to directly prove the absolute separability of $\hat\rho_V$;
here we shall only higlight that $\hat\rho_{KV} \equiv K \hat\rho_V K^\dagger$ is still separable:
\[
\hat\rho_{KV} \equiv K \hat\rho_V K^\dagger =
\frac{1}{8} \left(
\begin{array}{cccc}
 3 & 0 & 0 & 0 \\
 0 & 2 & 1 & 0 \\
 0 & 1 & 2 & 0 \\
 0 & 0 & 0 & 1 \\
\end{array}
\right)
\]
and the spectrum of $(\hat\rho_{KV})^{(A)} = (\hat\rho_{KV})^{(B)}$ is
$\{2 + \sqrt{2}, 2, 2, 2 - \sqrt{2} \}$, which is sufficient (in $2 \times 2$ dimensions)
for separability of $\hat\rho_{KV}$.

\section{Quantum teleportation}

Quantum teleportation is an interesting feature of entangled systems which allows to transmit
quantum information (a qubit) through a classical channel. This possibility relies on the fact
that <<\emph{we may cut a cake in different ways}>> \cite{thirring}, i.e.~the same composite system
may be factorized in different ways, as we have discussed at lenght.

Suppose two observers, Alice and Bob, can communicate, at first, through a q-channel.
Alice prepares the state $|\Psi^- \rangle$ of the Bell basis~\footnote{For the Bell basis see
{\bf Subs. \ref{sbell}}.} for a composite system ${\cal A} \otimes {\cal B}$
and sends one qubit to Bob. Later Alice recieves, through a q-channel ${\cal C}$ an unknown qubit
$| \psi \rangle = a |\uparrow \rangle + b |\downarrow \rangle$ which she wants to transmit to Bob via classical communication.
Alice knows that the state of the  system ${\cal C} \otimes {\cal A} \otimes {\cal B}$ is now:
\[
 \begin{split}
 | \psi \rangle \otimes | \Psi^- \rangle =
  &(a | \uparrow \rangle + b |\downarrow \rangle) \otimes 
  \frac{| \uparrow \downarrow \rangle - |\downarrow \uparrow \rangle}{\sqrt{2}} \\
 = &- \frac{1}{2} |\Psi^- \rangle \otimes (a | \uparrow \rangle + b |\downarrow \rangle)
    + \frac{1}{2} |\Phi^- \rangle \otimes (a | \downarrow \rangle + b |\uparrow \rangle) + \\
   &+ \frac{1}{2} |\Phi^+ \rangle \otimes (a | \downarrow \rangle - b |\uparrow \rangle)
    - \frac{1}{2} |\Psi^+ \rangle \otimes (a | \uparrow \rangle - b |\downarrow \rangle).
\end{split}
\]
The four terms of the sum, states of the total system, are all composed of:
\begin{itemize}
 \item a Bell state, over ${\cal C} \otimes {\cal A}$;
 \item a state that can be transformed into $| \psi \rangle$ with an unitary transformation, over ${\cal B}$,
which is in Bob's hands.
\end{itemize}
So now Alice just has to measure in which Bell state her subsystem is; the operation will leave
Bob's subsystem in a well-defined state, which Alice does not know, while she knows
what transformation Bob must perform, in order to obtain $| \psi \rangle$ on ${\cal B}$. Since there are
four possible transformations, she will need two classical bits to tell Bob which one to choose.

More generally speaking, let us notice that the hypothesis for ${\cal A}$, ${\cal B}$,
and ${\cal C}$ to be qubits is not essential, having just an illustrative purpose;
it is sufficient for the three systems to have the same dimension.

Furthermore, Alice can initially prepare on ${\cal A} \otimes {\cal B}$ whatever maximally
entangled state she wants. Her choice will define a bijective correspondence between basis
vectors (an \emph{isometry}) $J_{AB} : {\cal A} \longrightarrow {\cal B}$.
  In fact every maximally entangled state $| \phi \rangle$
on a $d \times d$ dimensional bipartite space can be written in the form
\[
| \phi \rangle = \frac{1}{\sqrt{d}}
\sum_{i=1}^d  | \phi_i^{(A)} \rangle \otimes | \phi_i^{(B)} \rangle \equiv
\sum_{i=1}^d  | \phi_i^{(A)} \rangle \otimes | J_{AB} (\phi_i^{(A)}) \rangle,
\]
Where the $\{|\phi_i^{(A)} \rangle\}$ and the $\{|\phi_i^{(B)} \rangle\}$ are bases for the respective
subspaces.
In the previous example
$$
|\Psi^- \rangle = \frac{| \uparrow \downarrow \rangle - |\downarrow \uparrow \rangle}{\sqrt{2}} \equiv
\frac{| \uparrow \rangle | J_{AB} (\uparrow) \rangle - |\downarrow \rangle | J_{AB} (\downarrow) \rangle}{\sqrt{2}}
$$
defines $J_{AB}$.

Another isometry $J_{CA}$ is defined when Alice performs the projective measurement over a basis
of maximally entangled states for ${\cal C} \otimes {\cal A}$.
The composition $J_{AB} \circ J_{CA} \equiv J_{CB} : {\cal C} \longrightarrow {\cal B}$ is what
Bob needs to know in order to decode his qubit.

%%%%%%%%%%%%%%%%
\section{Geometry of quantum states}

It can be useful and interesting to visualize the location of
entangled and separable states among the set of all physical states.

In Fig. \ref{fig:geom} the space of density matrices of a two-qubit system is plotted
with respect to the tensor products of Pauli matrices in the three directions.
All states lie within a tetrahedron whose vertices are the four Bell states (\ref{bell}).

\begin{figure}[th]
\centering
\includegraphics[width=12.6cm,height=11.4cm]{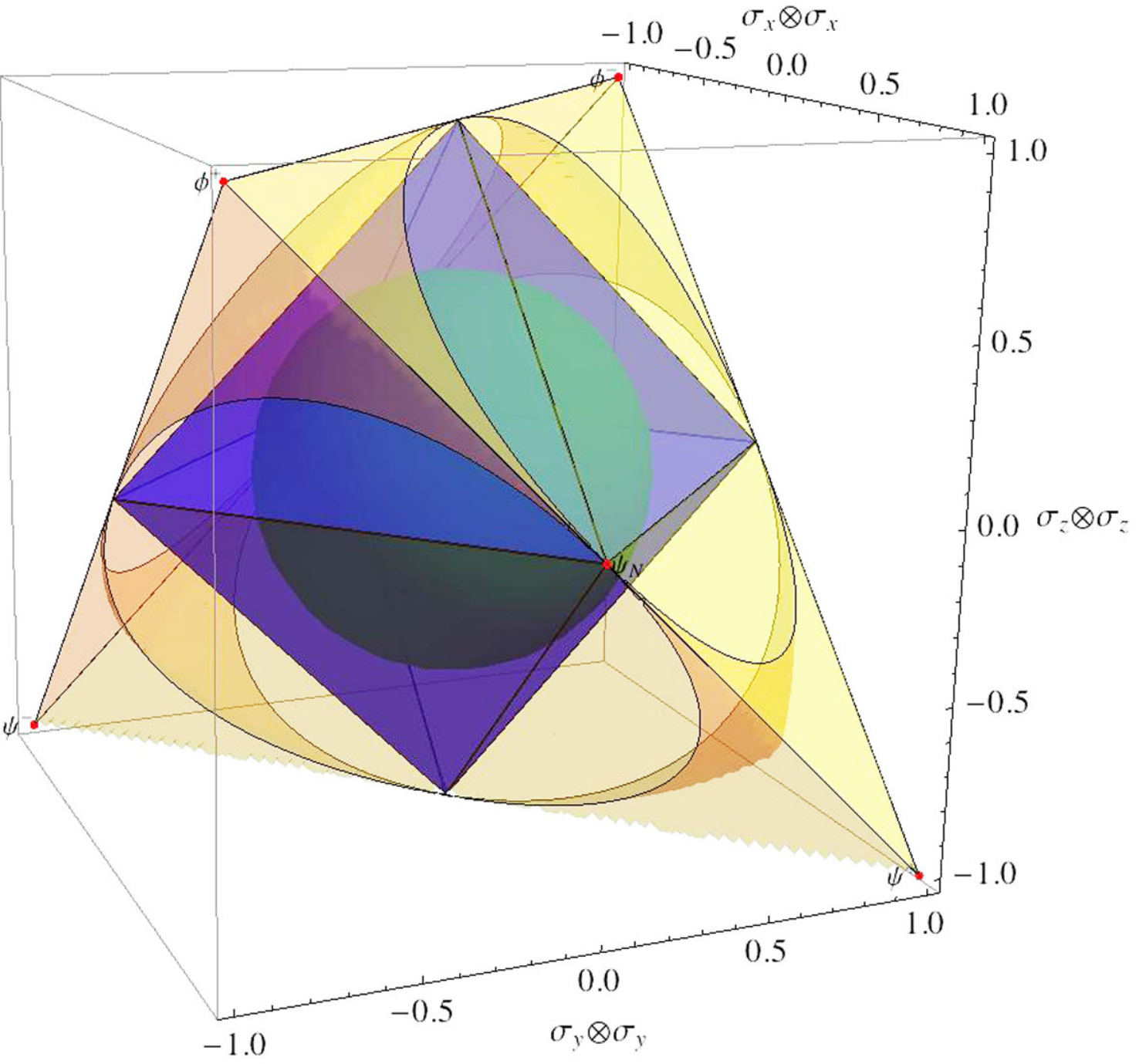}
\caption{\emph{Plot of physical states in $2 \times 2$ dimensions} \cite[Fig. 1]{thirring}.}
\label{fig:geom}
\end{figure}

Let us briefly show why: any density matrix can be expanded in the form
$$
\hat\rho = \frac{1}{4} ({\mathbb I} \otimes {\mathbb I} + c_x \hat\sigma_x \otimes \hat\sigma_x 
+ c_y \hat\sigma_y \otimes \hat\sigma_y + c_z \hat\sigma_z \otimes \hat\sigma_z).
$$
The eigenvalues of $\hat\rho$ can be expressed as four functions of the three parameters $c_x, c_y, c_z$;
setting the eigenvalues equal to zero, we find the equations
of the planes where the four faces of the tetrahedron lie.
The condition of positive semidefiniteness requires that all the
states lie inside the region enclosed within the four planes.

Inside the tetrahedron, the separable states form a double pyramid with its vertices lying on
the boundary of the tetrahedron itself, halfway between the Bell states; one of those vertices
is the state (\ref{rhomax}), which is maximally entangled just like the Bell states.
Actually, the double pyramid corresponds to the intersection of the tetrahedron and
its mirror image (with respect to the identity).

The KZ ball is set in the middle of the space, obviously contained in the double pyramid.
Its center is the normalized identity.

%%%%%%%%%%%%%%%%%%%%%
\chapter{Conclusions}
\begin{quotation}
\begin{flushright}
 ``We always have had a great deal of difficulty in understanding the world view  that quantum mechanics represents.
   [...] it takes a generation or two until it becomes obvious that there's no real problem.'' \\ 
 \vspace{0.2cm} \emph{R. Feynman}
\end{flushright}
\end{quotation}
\vspace{1cm}

We have here represented some of the main results obtained in the search for an answer to the question
``which states can be represented as entangled (or separable)?''.
It has been shown that, for the case of pure states, the situation is completely clear. This is in some way due to
the fact that, if we can describe a state by a wavefunction, we then have complete knowledge about it, at least
as much as the indeterministic nature of quantum mechanics allows us.

On the other hand, when a state is expressed in terms of a density operator, we generally lack information about the
non-local correlations this state may carry with respect to a certain factorization.
So, while it is always possible to represent it as separable, corresponding to considering all
quantum correlations as local, we cannot always be sure to find a factorization where
this correlation appears non-local.

Thus, except for low-dimensional cases, like the two qubits, for which more definitive conclusions have
been drawn, the answer to our question is far from having been given.
Efforts in this direction are various in recent literature, many of the works included in our bibliography
can be regarded as examples; also the strategies used are several: from the attempt to describe 
with the best accuracy the geometry of quantum states, to the study of entanglement in identical particles
systems or of entanglement measures, and many others. \\
The great interest of this subject is not only due to its fascinating theoretical aspects, but also due
to its importance from the strictly experimental or applications point of view, from
the physics of condensed matter to the data encryption. In recent years the first prototypes
of quantum hardwares and softwares have been developed and began to be commercialized.

%%%%%%%%%%%%%%%%%
\appendix

\chapter{Algebraic definitions and properties} \label{algebra}

In this appendix we recall some useful algebraic definitions and properties, which are used in the rest of the work.
Main references are \cite{bratteli} and \cite{caltagirone}.

\begin{definition}[Algebra]
 Let ${\cal A}$ be a vector space over the field ${\mathbb C}$ equipped with an inner operation
$\cdot : {\cal A} \times {\cal A} \longrightarrow {\cal A}$ which associates $x \cdot y$ (or simply $xy$)
to each pair $x,y \in {\cal A}$. ${\cal A}$ is an \emph{algebra} if this operation is distributive and associative, i.e.,
if for all $x, y, z \in {\cal A}$ and all $a, b \in {\mathbb C}$ holds:
\begin{itemize}
 \item $x \cdot (y \cdot z) = (x \cdot y) \cdot z$;
 \item $x \cdot (y + z) = x \cdot y + x \cdot z$;
 \item $(ab) (x \cdot y) = (ax) \cdot (by)$.
\end{itemize}
\end{definition}

\begin{definition}[Subalgebra]
 A subspace ${\cal B}$ of an algebra ${\cal A}$ which is also an algebra with respect to the same
operations of ${\cal A}$ is called a \emph{subalgebra} of ${\cal A}$.
Given any two subalgebras ${\cal A}$ and ${\cal B}$, the expression ${\cal A} \vee {\cal B}$ denotes the
\emph{minimal subalgebra} containing both ${\cal A}$ and ${\cal B}$.
\end{definition}

\begin{definition}[Identity]
 An \emph{identity} ${\mathbb I}$ of an algebra ${\cal A}$ is an element of ${\cal A}$ such that
$A = {\mathbb I} A = A {\mathbb I}$ for all $A \in {\cal A}$. Note that the identity is unique, but not every algebra has an identity.
\end{definition}

\begin{definition}[Ideal]
 A subspace ${\cal B}$ of an algebra ${\cal A}$ is called a \emph{left ideal} (respectively \emph{right ideal}) if $a \in {\cal A}$
and $b \in {\cal B}$ implies that $ab \in {\cal B}$ (respectively $ba \in {\cal B}$). An ideal which is both left and right is
called a \emph{two-sided ideal} (or simply an ideal). Note that each ideal is an algebra.
\end{definition}

\begin{definition}[Simple algebra]
 An algebra with identity is \emph{simple} if its only two-sided ideals are $\{ 0\}$ and the algebra itself,
i.e., it has no nontrivial two-sided ideals.
\end{definition}

\begin{definition}[Center]
 The center of an algebra ${\cal A}$ is the set of all those $c \in {\cal A}$ such that $ca = ac$ for every $a \in {\cal A}$.
\end{definition}

\begin{definition}[Central algebra]
An algebra is \emph{central} if its center consists only of the multiples of the identity.
\end{definition}

\begin{lemma} \label{full} 
 The full matrix algebra ${\cal M}_d$ is \emph{simple} and \emph{central}.
\end{lemma}

\begin{definition}[Centralizer, or commutant, of an algebra]
 Given a subalgebra ${\cal A}$ of an algebra ${\cal M}$, the \emph{commutant} ${\cal A}'$ of ${\cal A}$ in ${\cal M}$
is the space of all members of ${\cal M}$ commuting with every element of ${\cal A}$, i.e.
$$
{\cal A}' = \{ b \in {\cal M} | (\forall a \in {\cal A})(ab = ba) \}.
$$
Every commutant is an algebra.
\end{definition}

\begin{lemma}[Double centralizer] \label{double} 
Given a finite-dimensional simple central algebra ${\cal M}$ over an arbitrary field
and a simple subalgebra ${\cal A} \subseteq {\cal M}$,
the centralizer ${\cal A}'$ is simple, ${\cal A}'' = {\cal A}$, 
and $ \dim {\cal A}' \cdot \dim {\cal A} = \dim {\cal M}$.
\end{lemma}

%%%%%%%%%%%%%%%%
\chapter{Proofs} \label{proof}

In this appendix are presented, for completeness, the proofs of some statements from the rest of the work.

%%%%%%%%%%%%%%%%%%%%%%%%%%%%%%%%%%%%%%%%%%%
\section{Theorem \ref{TPS} {(Induced TPSs)}}

\emph{Statement}.
Consider the full matrix algebra ${\cal M}_d$ on the finite dimensional Hilbert space ${\cal H} = {\mathbb C}^d$
and two subalgebras ${\cal A}$ and ${\cal B}$ of ${\cal M}_d$, for which there hold:
\begin{description}
 \item[(i)] Independence: $[{\cal A},{\cal B}] = {0}$, i.e. $[a,b] = 0 \; \forall \; (a \in {\cal A}, b \in {\cal B})$;
 \item[(ii)] Completeness: ${\cal A} \otimes {\cal B} \cong {\cal A} \vee {\cal B} = {\cal M}_d$.
\end{description}
Then ${\cal A}$ and ${\cal B}$ induce a TPS on ${\mathbb C}^d$, i.e.~there exist two H-spaces ${\mathbb C}^k$ and ${\mathbb C}^l$,
$d = k \cdot l$, and an unitary mapping $U : {\mathbb C}^k \otimes {\mathbb C}^l \rightarrow {\mathbb C}^d$ such that
$$
{\cal A} = U({\cal M}_k \otimes {\mathbb I}_l) U^\dagger \mbox{\emph{ and }} {\cal B} = U({\mathbb I}_k \otimes {\cal M}_l) U^\dagger.
$$
In particular, ${\cal A}$ and ${\cal B}$ are isomorphic to ${\cal M}_k$ and ${\cal M}_l$ respectively.

\begin{proof}
The algebra ${\cal A}$ can be decomposed in a direct sum of simple subalgebras
${\cal A} = \oplus^n_{i=1} {\cal A}_i$, and each ${\cal A}_i$ is isomorphic to a full
matrix algebra ${\cal M}_{k_i}$, where the $k_i$ are positive integers uniquely determined by ${\cal A}$ up to
permutations \cite[Th. 11.2]{takesaki}. So we can write:
$$
{\cal A} \cong \bigoplus^n_{i=1} {\cal M}_{k_i}.
$$

Any representation (say $\Pi$) of ${\cal A}$ on ${\mathbb C}^d$ can be expressed as a direct sum of ${\cal M}_{k_i}$ times $l_i$, where
the $l_i$ are positive integers called \emph{multiplicities} (note that $\sum^n_{i=1} k_i \cdot l_i = d$):
$$
\Pi({\cal A}) = \bigoplus^n_{i=1} l_i {\cal M}_{k_i} = \bigoplus^n_{i=1} ({\mathbb I}_{l_i} \otimes {\cal M}_{k_i}).
$$
$\Pi$ is uniquely determined by the multiplicities up to a unitary equivalence, i.e. given any two representations $\Pi_1$ and
$\Pi_2$ with the same multiplicities there will always exist an unitary operator $U: \bigoplus^n_{i=1} ({\mathbb C}^{k_i} \otimes
{\mathbb C}^{l_i}) \rightarrow {\mathbb C}^d$ such that $\Pi_1 = U \Pi_2 U^\dagger$.
In particular, considering ${\cal A}$ as a trivial representation of itself with multiplicities $\{a_1, \cdots, a_n\}$, 
for ${\cal A}$ and its centralizer ${\cal A}'$ we can write
\begin{eqnarray*}
{\cal A} &=& U(\bigoplus^n_{i=1} {\mathbb I}_{a_i} \otimes {\cal M}_{k_i}) U^\dagger, \\
{\cal A}' &=& U(\bigoplus^n_{i=1} {\cal M}_{a_i} \otimes {\mathbb I}_{k_i}) U^\dagger \\
\Rightarrow \qquad {\cal A}\vee {\cal A}' &=& U(\bigoplus^n_{i=1} {\cal M}_{a_i} \otimes {\cal M}_{k_i}) U^\dagger \subseteq {\cal M}_d, \\
\dim {\cal A} \vee {\cal A}' &=& \sum^n_{i=1} a^2_i \cdot k^2_i \leq d^2. 
\end{eqnarray*}

By conditions {\bf (i)} and {\bf (ii)} it follows respectively ${\cal B} \subseteq {\cal A}'$ and $\dim {\cal A} \vee {\cal B} = d^2$.
Putting it all together we have
\begin{equation} \label{dimTPS}
d^2 = (\sum^n_{i=1} a_i \cdot k_i)^2 = \dim {\cal A} \vee {\cal B} \leq \dim {\cal A} \vee {\cal A}' = \sum^n_{i=1} a^2_i \cdot k^2_i \leq d^2.
\end{equation}
Being $a_i$ and $k_i$ positive numbers, \ref{dimTPS} implies $n = 1$. Thus ${\cal A}$ is simple and central
(it is isomorphic to ${\mathbb I}_a \otimes {\cal M}_k$) and we can apply {\bf Lemma \ref{double}}:
$$
\dim {\cal A}\cdot  \dim {\cal A}' = \dim {\cal M}_d.
$$
By {\bf (ii)} it holds
$$
\dim {\cal A}\cdot  \dim {\cal B} = \dim {\cal M}_d,
$$
so that
$$
\dim {\cal A}' = \dim {\cal B}.
$$
This equality, together with ${\cal B} \subseteq {\cal A}'$, implies 
${\cal B} = {\cal A}' = U({\cal M}_a \otimes {\mathbb I}_k) U^\dagger.$
\end{proof}

%%%%%%%%%%%%%%%%%%%%%%%%%%%%%%%%%%%%%%%%%%%%%%%%%%%%%%%
\section{Theorem \ref{tailor} {(Tailored observables)}}

\emph{Statement}.
Let ${\cal H} = {\mathbb C}^d$ be a H-space with an orthonormal basis $\{ |i \rangle \}$, and take $k, l \in {\mathbb N}$ such that $d = k \cdot l.$
Then, for every (pure) state $| \psi \rangle = \sum^d_{i=1} c_i|i \rangle$ and every $\lambda_1, \cdots ,\lambda_{min \{k,l\}} \in {\mathbb C}$
with $\sum^d_{i=1} |c_i|^2 = \sum^{min \{k,l\}}_{i=1} |\lambda_i|^2$, there exist algebras ${\cal A}$ and ${\cal B}$ satisfying the conditions
of {\bf Theorem \ref{TPS}}, and a unitary operator $U$ such that
$$
| \psi \rangle = U \sum^{min \{k,l\}}_{i=1} \lambda_i |i \rangle_A |i \rangle_B,
$$
with $\{|i \rangle_A \}$ and $\{|i \rangle_B \}$ orthonormal bases for H-spaces
${\cal H}_A = {\mathbb C}^k$ and ${\cal H}_B = {\mathbb C}^l$, respectively.

\begin{proof}
We assume $k \leq l$ and we define the H-space ${\cal H}' \equiv {\mathbb C}^k \otimes {\mathbb C}^l$. Since $\dim {\cal H} = \dim {\cal H}'$,
we have ${\cal H} \cong {\cal H}'$.
Subspaces ${\mathbb C}^k$ and ${\mathbb C}^l$ will have local algebras of observables ${\cal A}'$ and ${\cal B}'$ which,
thanks to {\bf Lemma \ref{su2}}, are generated by the sets of operators $\{ S^{A'}_x, S^{A'}_y, S^{A'}_z \}$ and
$\{ S^{B'}_x, S^{B'}_y, S^{B'}_z \}$ respectively; the $S_{x,y,z}$ are the well-known abstract generators of $\mathfrak{su}(2)$
fulfilling the commutation relation $[S_i,S_j] = \mbox{i} \hbar \epsilon_{ijk} S_k$.~\footnote{In physics they usually represent
the angular momentum operators along the three axes.}

As a basis for ${\mathbb C}^k$ (${\mathbb C}^l$) we can choose the eigenvectors of $S^{A'}_z$ ($S^{B'}_z$),
$\{|i \rangle_{A'}\}^k_{i=1}$ and $\{|i \rangle_{B'}\}^l_{i=1}$, respectively.
An arbitrarily entangled state $|\phi \rangle$ on ${\cal H}'$ may then be written in its Schmidt form as
$$
|\phi \rangle = \sum^k_{i=1} \lambda_i |i \rangle_{A'} |i \rangle_{B'},
$$
where the $\lambda_i$ are the Schmidt coefficients (positive real numbers), containing all the information about the amount of entanglement,
and the sum over $i$ is limited by the dimension of the smallest of the subspaces. Note that, as far as both $| \psi \rangle$
and $| \phi \rangle$ are normalized, we have $\sum^k_{i=1} |\lambda_i|^2 = \sum^d_{i=1} |c_i|^2 = 1$.

Now, using the Gram-Schmidt procedure \cite[p. 511]{lang},
we can construct new bases of ${\cal H}$ and ${\cal H}'$ containing respectively $| \psi \rangle$ and $| \phi \rangle$:
$\{ | \psi_1 \rangle = | \psi \rangle, \cdots, | \psi_d \rangle \}$ and
$\{ | \phi_1 \rangle = | \phi \rangle, \cdots, | \phi_{k \cdot l = d} \rangle \}$.

Since ${\cal H} \cong {\cal H}'$ we can always find a unitary operator $U: {\mathbb C}^k \otimes
{\mathbb C}^l \rightarrow {\mathbb C}^d$ mapping ${\cal H}'$ into ${\cal H}$ such that:
$$
U |\phi_i \rangle = |\psi_i \rangle, \; ( i = 1, \cdots, d).
$$
It is immediate to see that the subalgebras ${\cal A}$ and ${\cal B}$ acting on ${\cal H}$ can be written as
$$
{\cal A} = U {\cal A}' U^{\dagger} = U ({\cal M}_k \otimes {\mathbb I}_l) U^{\dagger}, \quad
{\cal B} = U {\cal B}' U^{\dagger} = U ({\mathbb I}_k \otimes {\cal M}_l) U^{\dagger},
$$
and are generated respectively by the operators 
$$
U (S^{A'}_i \otimes {\mathbb I}_l) U^{\dagger} \mbox{ and } U ({\mathbb I}_k \otimes S^{B'}_i) U^{\dagger}, \; (i = x,y,z),
$$
i.e.~by representations of $\mathfrak{su}(2)$ on ${\mathbb C}^d$.

Thus, because of {\bf Lemma \ref{su2}}, ${\cal A}$ and ${\cal B}$ fulfill the conditions from {\bf Theorem \ref{TPS}}.
\end{proof}

%%%%%%%%%%%%%%%%%%%%%%%%%%%%%%%%%%%%%
\section{Optimal entanglement witness} \label{ew}

Here we prove that the operator defined in (\ref{opt}) is an optimal entanglement witness. The original
theorem, known as \emph{Bertlmann-Narnhofer-Thirring theorem}, was first demonstrated in
\cite[Sect. III]{bertlmann}.

\emph{Statement}.
Given any entangled state $\hat\rho_e$, and $\hat\rho_0$ the separable state for which
$\delta (\hat\rho_e , \hat\rho_0)$ is minimum, consider the operator:
$$
E \equiv \frac{\hat\rho_0 - \hat\rho_e - \langle \hat\rho_0 | \hat\rho_0 - \hat\rho_e \rangle {\mathbb I}}
                    {\lVert \hat\rho_0 - \hat\rho_e \rVert}.
$$
Then $E$ is an optimal EW with respect to $\hat\rho_e$, i.e.~it has the three
properties (\ref{ew1}), (\ref{ew2}), (\ref{ew3}).

\begin{proof}
 Let us set, without any loss of generality, $\lVert \hat\rho_0 - \hat\rho_e \rVert = 1$. We will thus have,
for any state $\hat\rho$:
\[
\begin{split}
\langle \hat\rho | E \rangle &= \langle \hat\rho | \hat\rho_0 - \hat\rho_e \rangle - \langle \hat\rho_0 | \hat\rho_0 - \hat\rho_e \rangle
\langle \hat\rho | {\mathbb I} \rangle \\
 &= \langle \hat\rho | \hat\rho_0 - \hat\rho_e \rangle - \langle \hat\rho_0 | \hat\rho_0 - \hat\rho_e \rangle \\
 &= \langle \hat\rho - \hat\rho_0 | \hat\rho_0 - \hat\rho_e \rangle,
\end{split}
\]
since $\langle \hat\rho | {\mathbb I} \rangle = \mbox{Tr} (\hat\rho^\dagger {\mathbb I}) = \mbox{Tr} \hat\rho = 1$.

It will then clearly hold:
\begin{gather*}
 \langle \hat\rho_e | E \rangle = \langle \hat\rho_e - \hat\rho_0 | \hat\rho_0 - \hat\rho_e \rangle
= - \lVert \hat\rho_e - \hat\rho_0 \rVert^2  \leq 0; \\
\langle \hat\rho_0 | E \rangle = 0;
\end{gather*}
which prove respectively (\ref{ew1}) and (\ref{ew3}).

Now we pass to (\ref{ew2}): taken any separable $\hat\rho_s$ (except, of course, $\hat\rho_0$),
we can surely decompose $|\hat\rho_s - \hat\rho_0 \rangle$
orthogonally with respect to $| \hat\rho_0 - \hat\rho_e \rangle$:
$$
|\hat\rho_s - \hat\rho_0 \rangle \equiv |v_\perp \rangle + |v_\parallel \rangle,
$$
with
$$
\langle v_\perp | \hat\rho_0 - \hat\rho_e \rangle = 0, \quad | v_\parallel \rangle = \alpha |\hat\rho_0 - \hat\rho_e \rangle,
$$
where $\alpha$ is a real number.

It must hold $\lVert \hat\rho_0 - \hat\rho_e \rVert \leq \lVert v_\parallel \rVert$, i.e. $\alpha \geq 1$, as we prove 
by way of contradiction (See Fig. \ref{fig:ew} for illustration).

\begin{figure}[th]
\centering
\includegraphics[width=12.6cm,height=9.5cm]{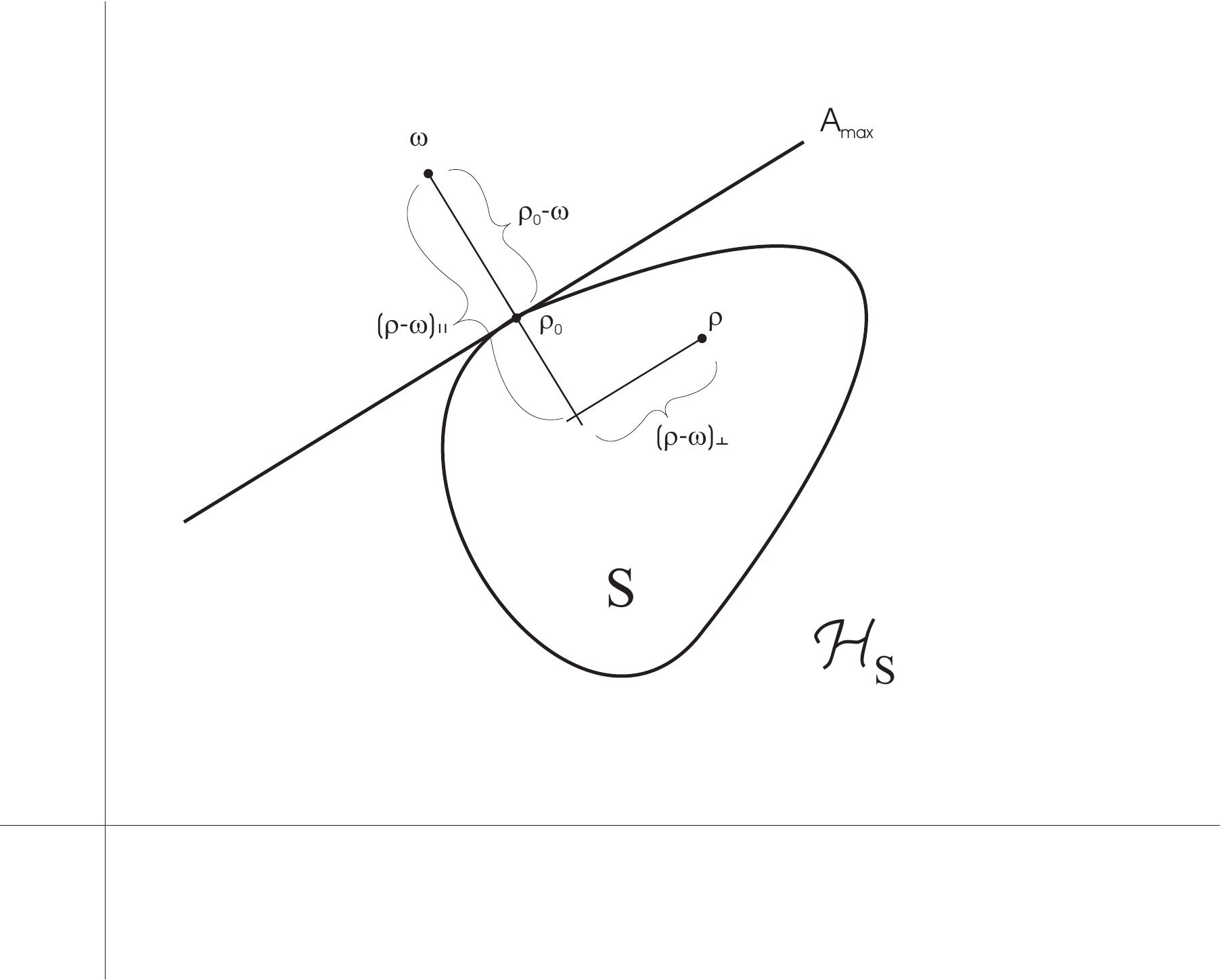}
\caption{\emph{Illustration of the theorem in {\bf Sect. \ref{ew}}. Here $\omega$ denotes $\hat\rho_e$, $\rho$ denotes $\hat\rho_s$,
$S$ is the set of separable states and ${\cal H}_s$ is the whole algebra of observables.
$A_{max}$ is $\partial C$} \cite[Fig. 2]{lewenstein}.}
\label{fig:ew}
\end{figure}

Consider the two sets $C \equiv \{ \hat\rho : \lVert (\hat\rho - \hat\rho_e)_\parallel \rVert
< \lVert (\hat\rho_0 - \hat\rho_e)_\parallel \rVert \}$,
with its boundary $\partial C = \{ \hat\rho : \lVert (\hat\rho - \hat\rho_e)_\parallel \rVert
= \lVert (\hat\rho_0 - \hat\rho_e)_\parallel \rVert \}$,
and $B_e = B(\hat\rho_e, \delta(\hat\rho_e,\hat\rho_0))$.
\newpage
Notice that:
\begin{itemize}
 \item  $\lVert (\hat\rho_0 - \hat\rho_e)_\parallel \rVert = \lVert \hat\rho_0 - \hat\rho_e \rVert$;
 \item  $\hat\rho_0$ is a boudary point for $B_e$, and $\partial C$ is the tangent plane
to $B_e$ in $\hat\rho_0$; in fact, if $\hat\rho_1$ and $\hat\rho_2$ belong to $\partial C$, it holds:
\[
 \begin{split}
  \langle \hat\rho_1 - \hat\rho_2 | \hat\rho_0 - \hat\rho_e \rangle &=
  \langle (\hat\rho_1 - \hat\rho_2)_\parallel | \hat\rho_0 - \hat\rho_e \rangle \\
  &= \langle (\hat\rho_1 - \hat\rho_e + \hat\rho_e - \hat\rho_2)_\parallel | \hat\rho_0 - \hat\rho_e \rangle \\
  &= \langle (\hat\rho_1 - \hat\rho_e)_\parallel - (\hat\rho_2 - \hat\rho_e)_\parallel | \hat\rho_0 - \hat\rho_e \rangle = 0,\\
\end{split}
\]
by assumption.

 \item $\partial C \nsubseteq C$.
\end{itemize}

 Since both $C$ and the set of separable states are convex,
then, if $C$ contains a certain separable $\hat\rho_C$, it will also contain every $\hat\rho_\lambda$
belonging to $L \equiv \{ \hat\rho : \hat\rho = (1 - \lambda) \hat\rho_C + \lambda \hat\rho_0,
\lambda \in [0,1[\; \}$, and $\hat\rho_\lambda$
will be separable. So, when $\lambda \rightarrow 1$,
$\hat\rho_\lambda$ gets arbitrarily close to $\hat\rho_0$.

The angle between $|\hat\rho_\lambda - \hat\rho_0 \rangle$ and $| \hat\rho_0 - \hat\rho_e \rangle$ obviously
does not depend on the parameter $\lambda$:
\[
 \begin{split}
 \frac{\langle \hat\rho_\lambda - \hat\rho_0| \hat\rho_0 - \hat\rho_e \rangle}
{\lVert \hat\rho_\lambda - \hat\rho_0\rVert \lVert \hat\rho_0 - \hat\rho_e \rVert}
&= \frac{(1-\lambda)\langle \hat\rho_C - \hat\rho_0| \hat\rho_0 - \hat\rho_e \rangle}
{(1-\lambda)\lVert \hat\rho_C - \hat\rho_0\rVert \lVert \hat\rho_0 - \hat\rho_e \rVert} \\
&= \frac{\langle \hat\rho_C - \hat\rho_0| \hat\rho_0 - \hat\rho_e \rangle}
{\lVert \hat\rho_C - \hat\rho_0\rVert \lVert \hat\rho_0 - \hat\rho_e \rVert};
\end{split}
\]
this means that also the angle between $L$ and $\partial C$ is fixed, being $\partial C$ orthogonal
to $| \hat\rho_0 - \hat\rho_e \rangle$.
Since it holds $\hat\rho_C \in C \Longrightarrow \hat\rho_C \notin \partial C$, then $L$ does not lie on
$\partial C$: the straight line on which $L$ lies is secant with
respect to $B_e$.
So, while $\hat\rho_\lambda \rightarrow \hat\rho_0$, there are only two possibilities:
\begin{itemize}
 \item either $\hat\rho_\lambda$ lies inside $B_e$ for every $\lambda$ greater then a certain $\lambda_0$ (it ``reaches from inside'');
 \item or $\hat\rho_\lambda$ never lies inside $B_e$ (it ``reaches from outside'').
\end{itemize}
But the second choice is impossible, because $\hat\rho_0$ is a boundary point for $C$ too, and $B_e \subseteq C$,
so if $\hat\rho_\lambda$ never lies inside $B_e$, it also never lies inside $C$, which is a contradiction.
There must then be some separable $\hat\rho_\lambda$ which lies inside $B_e$, i.e. $\delta (\hat\rho_\lambda ,\hat\rho_e )
< \delta (\hat\rho_0 , \hat\rho_e)$, and this is a contradiction too, since $\delta (\hat\rho_0 , \hat\rho_e)$
is assumed to be minimal.

So we have finally shown that $C$ cannot contain any separable $\hat\rho_s$, i.e. $\lVert \hat\rho_0 - \hat\rho_e \rVert
\leq \lVert (\hat\rho_s)_\parallel \rVert$.

We can now return to our statement: 
$|\hat\rho_s - \hat\rho_0 \rangle$ is in particular a separable state, so it holds, 
as announced, $\alpha \geq 1$. We actually have proved only $| \alpha | \geq 1$, but the set of separable
states must contain $|\hat\rho_s - \hat\rho_0 \rangle$ (the case $\alpha = 1$) and it must be convex,
so the parameter cannot ``jump'' to negative values.

Concluding, we have:
\[
 \begin{split}
\langle \hat\rho_s | E \rangle = \langle \hat\rho_s - \hat\rho_0 | \hat\rho_0 - \hat\rho_e \rangle
= \alpha \langle \hat\rho_0 - \hat\rho_e | \hat\rho_0 - \hat\rho_e \rangle
= \alpha \lVert \hat\rho_0 - \hat\rho_e \rVert^2 \geq 0,
 \end{split}
\]
which is the thesis.	
\end{proof}

%%%%%%%%%%%%%%%%%%%%%%%%%%
%\chapter{sandbox and notes}

%%%%%%%%%%%%%%%%
\cleardoublepage
%\phantomsection

\end{document}